\documentclass[%
reprint,
superscriptaddress,
amsmath,amssymb,
aps,
pra,
longbibliography,
]{revtex4-2}

\usepackage{amsthm}
\usepackage{newtxmath}
\usepackage{booktabs}
\usepackage{float}
\usepackage{graphicx}% Include figure files
\usepackage{dcolumn}% Align table columns on decimal point
\usepackage{bm}% bold math
\usepackage{hyperref}% add hypertext capabilities
\usepackage{physics}
\usepackage{diagbox}
\usepackage[caption=false, justification=centerlast]{subfig}
\hypersetup{colorlinks=true, citecolor=blue, urlcolor=blue, linkcolor=blue}
\usepackage{mathtools}
\usepackage{subfiles}
\usepackage{tabularx}
\usepackage{qcircuit}
\usepackage{float}
\makeatletter
\let\newfloat\newfloat@ltx
\makeatother
\usepackage{algorithm}
\usepackage{algorithmicx}
\usepackage{algpseudocode}
\usepackage{multirow}
\usepackage[english]{babel}

\newtheorem{definition}{Definition}
\newtheorem{theorem}{Theorem}
\newtheorem{corollary}{Corollary}
\newtheorem{lemma}{Lemma}
\newtheorem{proposition}{Proposition}

\begin{document}
\title{Optimal Number-Conserved Linear Encoding for Practical Fermionic Simulation}%Unleashing Quantum Simulation Advantages: Hamiltonian Subspace Encoding for Resource Efficient Quantum Simulations}% Force line breaks with \\

\author{M. H. Cheng}
\altaffiliation{These authors contributed equally to this work}
\affiliation{Department of Physics Blackett Laboratory, Imperial College London, London, SW7 2AZ, United Kingdom}
\affiliation{Fraunhofer Institute for Industrial Mathematics, Fraunhofer-Platz 1, 67663 Kaiserslautern, Germany}

\author{Yu-Cheng Chen}
\altaffiliation{These authors contributed equally to this work}
\affiliation{Department of Mechanical Engineering, City University of Hong Kong, Kowloon, Hong Kong SAR 999077, China}
\affiliation{Hon Hai (Foxconn) Research Institute, Taipei, Taiwan}

\author{Qian Wang}
\affiliation{Department of Mechanical Engineering, City University of Hong Kong, Kowloon, Hong Kong SAR 999077, China}

\author{V.  Bartsch}
\affiliation{Fraunhofer Center for Maritime Logistics, Blohmstrasse 32, 21079 Hamburg, Germany}

\author{A. C. Medina}
\affiliation{Fraunhofer Institute for Industrial Mathematics, Fraunhofer-Platz 1, 67663 Kaiserslautern, Germany}

\author{M. S. Kim}
\email{m.kim@imperial.ac.uk}
\affiliation{Department of Physics Blackett Laboratory, Imperial College London, London, SW7 2AZ, United Kingdom}
\email{m.kim@imperial.ac.uk}

\author{Alice Hu}
\email{alicehu@cityu.edu.hk}
\affiliation{Department of Mechanical Engineering, City University of Hong Kong, Kowloon, Hong Kong SAR 999077, China}
\affiliation{Department of Materials Science and Engineering, City University of Hong Kong, Kowloon, Hong Kong SAR 999077, China}

\author{Min-Hsiu Hsieh}
\email{min-hsiu.hsieh@foxconn.com}
\affiliation{Hon Hai (Foxconn) Research Institute, Taipei, Taiwan}

\newcommand{\cheng}[1]{\textcolor{black}{#1}}
\newcommand{\acm}[1]{\textcolor{red}{#1}}

\begin{abstract}%

Number-conserved subspace encoding reduces resources needed for quantum simulations, but scalable complexity trade-off bounds for $M$ modes and $N$ particles with $\mathcal{O}(N\log M)$ qubits have remained unknown. We study qubit-gate-measurement trade-offs through the lens of classical/quantum error correction complexity, and develop a framework of fermionic gate and measurement complexity based on encoder and decoder complexities appeared in error correction framework. We demonstrate optimal encoding with random classical parity check code and propose the Fermionic Expectation Decoder for scalable probability decoding in $\mathcal{O}(M^4)$ bases. The protocol is tested with variational quantum eigensolver on LiH in the STO-3G and 6-31G basis, and $\text{H}_2$ potential energy curve in the 6-311G* basis.

%We show that such protocol has qubit-to-measurement basis trade-off of $\lceil 2N\log M\rceil$ and $\mathcal{O}(M^4)$ respectively where $M$ is the number of modes and $N$ the number of electrons. For classical pre and post-processing, we construct the polynomial time Fermionic Expectation Decoder to compute the fermionic observable and demonstrate the optimal encoder using the Randomized Linear Encoder algorithm. This protocol is tested with variational quantum eigensolver on $\text{LiH}$ in the STO-3G and 6-31G basis with a hardware-efficient ansatz at a bond length of 2.5 Å, and the potential energy curve of $\text{H}_2$ energy in the 6-311G* basis.
%A concluding remark is important.

%in this second part of the result, 

%\cheng{You must mention that your decoding scheme is useful for any forms of encodings.}

%here you will need to talk about two points. 
\end{abstract}
\maketitle
%\cheng{1. first, just look at what the referees have to say about the correction.}
%\cheng{2. second, we need to include the new result, like minimalistic.}
%\cheng{3. Then, try to smooth out the whole work. }

%\cheng{something like this needed to be mentioned: All calculations were performed using the open-source software packages OpenFermion and Psi444,45. }
%\cheng{a much more significant emphasis on the gate complexity with respect to compression must be emphasized in this work.}
%\cheng{Detail comments: the significance of this work can be seen from the existence of polynomial time encoder and decoder, and their optimal bounds. Who are the audiences? it is from the quantum chemistry community, which will need to use our code.}

%\cheng{Next, we need to discuss about the role of LDPC code and the efficiency of decoding.}

%\cheng{We could take reference to the writing style of the QEE paper.}
%\cheng{lets write down the comments in the first referee:
%\begin{itemize}
    
%    \item In Eqs. (D10) and (E1), the ranges of the sums should be given.
%    \item p. 4 center right: "We have also realized that linear encoding can
%avoid barren plateaus contributed by the unphysical states." remove this claim
%    \item improve theorem 2: working on it. Also, we should mention that for restricted shell simulation this will reduced to O(N) measurement basis. 
%\end{itemize}
%}

\section{Introduction}

Simulating many-body fermionic systems is an essential tool in various fields like quantum chemistry, condensed matter, and high energy physics \cite{Bauer_2020, kitaev2009topological, bauer2022quantum}, promising better designs for batteries and drugs, new quantum technologies, and discoveries in the exotic nature of matter \cite{ho_promise_2018, santagati2023drug, Zhang_2017}. However, accurate fermionic simulation on classical computers requires exponential computational resources \cite{feynman2018simulating}. While a quantum computer is expected to overcome this limit, it remains a central challenge to demonstrate quantum advantages in this domain.

\begin{table*}[!htb]
    \centering
    \begin{tabular}{|c|c|c|c|c|c|c|}
        \hline
        & Jordan Wigner \cite{Jordan1928wi}& Segment \cite{Steudtner_2018} &Graph-Based \cite{bravyi2017tapering} & Polylog \cite{kirby2022second} & QEE \cite{shee2022qubit} & Our Work\\
        \hline
        Qubit cost& $M$ & $M-\frac{M}{2N}$ & $M-\frac{M}{N}$& $\mathcal{O}(N^2(\log M)^4)$ & $\lceil N\log M\rceil$ &$\lceil 2N\log M\rceil$ \\
        \hline
        Measurements& $\mathcal{O}(M^3)$ \cite{huggins_efficient_2021}&$\mathcal{O}(M^4)$ &$\mathcal{O}(M^4)$ & $\mathcal{O}(M^4)$& $\mathcal{O}(M^N)$& $\mathcal{O}(M^4)$\\
        \hline
    \end{tabular}
    %this caption clearly requires 
    \caption{This table compares the qubit-to-measurement trade-off of the existing encoding schemes at the regime $M>>N$. \cheng{We also provide the scaling of measurement bases for the listed linear encoders contributed from the FED algorithm.}}
    \label{tab: resource comparison}
\end{table*}

Electronic structure Hamiltonian encoding maps fermionic operators from the Fock space to qubit operators in the qubit space. The standard approaches such as the Jordan-Wigner, Parity, and Bravyi-Kitaev transformations \cite{Jordan1928wi, seeley2012bravyi} encode $M$ fermionic modes to $M$ qubits, regardless of the number of electrons $N$. However, since all fermionic Hamiltonians must conserve fermions, interesting physics depends only on the number-conserved subspace and takes only $\lceil N\log M\rceil$ qubits to encode in the precision limit $M>>N$ \cite{Shepherd_2012, Gruneis_2013}. Logarithmic scaling in $M$ enormously relieves the qubit resource requirements for simulation in the high precision limit, and it could serve as a powerful tool in the practical implementation of quantum simulation on near-term and fault-tolerant quantum computing (FTQC) devices \cite{Bharti_2022}. 

%you need to first motivate it here.
Another important aspect of fermionic encoding is the measurement scalability of the two-electron reduced density matrix (2-RDM). Unlike classical encoding, quantum encoding changes the measurement complexity in two ways. First, non-linear methods such as the projector approach \cite{Yen_2019} and
the Qubit Efficient Encoding (QEE) \cite{shee2022qubit} have attained the $\lceil N\log M \rceil$ qubit scaling via a direct fermionic state to qubit state encoding. However, it is plagued with unscalable encoding and measurement basis preparation cost of $\mathcal{O}(M^N)$, limiting its applications to fermionic systems with few electrons. Even if a specific grouping is found to mitigate the basis preparation cost, it will still suffer from an unbounded number of commuting Pauli strings.

Linear compression resolves the measurement and decoding cost. All linear compression methods, such as the segment code \cite{Steudtner_2018} and graph-based encoding \cite{bravyi2017tapering}, exhibit polynomial measurement costs with respect to $M$. However, they fail to efficiently encode fermionic modes and large number of electrons. The current best linear encoding demonstrates poly-logarithmic $\mathcal{O}(N^2(\log M)^4)$ qubit scaling \cite{kirby2022second} that overcomes the scaling in the previous works, but still is limited to molecules requiring large basis set not accessible to near-term devices. The optimal compression rate in linear compression remains unknown.

%This work exploits the correspondence between linear compression and its dual code and reformulates number-conserved fermionic encoding in the language of classical error correction code . 
%should mention decoding sign. You forget to mention the measurement complexity. 
% the order of the importance is again emphasize here. 
We identify a natural complexity duality between linear compression for quantum simulation (QS) and quantum error correction (QEC) framework across qubit, gate and measurement encoding-decoding process. Subsequently, we develop a framework of fermionic gate complexity based on quantum decoder, and fermionic measurement complexity based on classical decoder. The latter result leads to the Fermionic Expectation Decoding (FED) algorithm which classically 'error' decodes \cheng{expectation value from compressed fermionic probability distribution (encoded Hamiltonian) for arbitrary linear encoding schemes.} By avoiding calculation in the encoded space, FED allows for polynomial time classical post-processing with $\mathcal{O}(M^4)$ measurement basis, compared to JW $\mathcal{O}(M^3)$ measurement basis scaling \cite{huggins_efficient_2021}. Meanwhile, using parity check code of error correction code (CEC) that obeys the Gilbert-Varshamov (GV) bound \cite{bravyi2017tapering, Gilbert_1952, Varshamov_1957}, we show that optimal linear encoder has a qubit resource upper bound of $2N\log M$, demonstrating a encoder-decoder protocol with the optimal qubit-to-measurement trade-off. Polynomial-time encoders and decoders can be found in examples such as the geometric Goppa code \cite{Goppa_construction, Goppa_decoder} and Gallager's low density parity check (LDPC) code \cite{Gallager, List_Decoding}. Our protocol resolves scalability issue in non-linear encodings in both the encoding and decoding process via parity check code, polynomial measurement bases and parity check decoder for fermionic observable post-processing. \cheng{Compatibility with other linear encoding schemes allows us to leverage the complexity of classical decoder.} As a result, this scheme finds board applications for encoding fermionic systems in Noisy Intermediate Scale Quantum (NISQ) or FTQC quantum algorithms such as qubit Adaptive Derivative-Assembled Pseudo-Trotter ansatz Variational Quantum Eigensolver (qubit ADAPT-VQE) \cite{Qubit-Adapt-VQE} and quantum signal processing \cite{Low_Cheung}.

%is this the part where you would like to put. Probably need to elaborate on the noise results.
We benchmark \cheng{Gilbert bounded} linear encodings with a Randomized Linear Encoder (RLE) and a classical look-up table decoder to encode large-mode problems with a bounded number of electrons. RLE encoded circuits adhere to the $2N\log M$ qubit bound, surpassing the qubit compression rates of the segment code and graph-based encoding for all $N$ and $M$ . The encoding is tested using the variational quantum eigensolver (VQE) on LiH in the STO-3G and 6-31G basis at a bond length of 2.5 Å with the hardware-efficient ansatz (HEA) \cite{kandala_hardware-efficient_2017, zeng2023quantum}, and $\text{H}_2$ in the 6-311G* basis. The encoded LiH noiseless simulation achieves chemical accuracy using 30 CNOT gates with 5 layers, compared to the \cheng{circuit in the JW basis}, which does not converge even after 6 layers. We also compare both cases with the IBM Sherbrooke noise model. $\text{H}_2$ simulation achieves chemical accuracy across all bond lengths using only 66 local CNOT gates.

%this part basically requires a complete revision.
% trying to submit this as soon as possible. 

%Whatever it is, now it all rests on luck.

%For instance, our method can encode a $2$-electron problem with $580$ fermionic modes within $30$ qubits, outperforming previous ab initio methods capable of encoding at most $59$ modes with the same electron number and qubit resources \cite{bravyi2017tapering}. 

%need to talk about the experiment here. The finalized version should be compared to the QEE paper.
%should we even mention the experimental results here?

%working on the story telling. 
%\textit{General Framework}\textemdash
\section{Operator Framework}

%why do you care about fermionic operator representation. 
Given a second-quantized electronic structure Hamiltonian:
\begin{equation}
    \hat{H} = \sum_{ij}h_{ij}\hat{a}^\dagger_i\hat{a}_j+\sum_{ijkl}g_{ijkl}\hat{a}^\dagger_i\hat{a}^\dagger_j\hat{a}_k\hat{a}_l,
    \label{eqn: hamiltonian}
\end{equation}
qubit operator representation of the excitation operators $\hat{a}^\dagger_i\hat{a}_j$ and $\hat{a}^\dagger_i\hat{a}^\dagger_j\hat{a}_k\hat{a}_l$, collectively denoted as $\hat{O}:=\hat{a}^\dagger_i \hat{a}_j, \hat{a}^\dagger_i\hat{a}^\dagger_j\hat{a}_k\hat{a}_l$, determines qubit, gate, and measurement complexities of quantum simulation. \cheng{In what follows, we review the framework and complexity of fermionic-to-qubit mapping, drawing connections to linear compression of number-conserved particle basis.}

% this part requires better motivations

%this part requires a further correction. % writing require further improvements.
Number-conserved bitstrings have fixed Hamming weight that can be compressed via a binary linear map $\mathbf{G}: \mathbb{F}^M_2\rightarrow \mathbb{F}^Q_2$. A fermionic/spin operator $\hat{O}$ can be represented with three types of binary vectors $\Vec{a}, \Vec{b}, \Vec{c}\in \mathbb{F}^M_2$. 
\cheng{In terms of qubit operator, these vectors respectively corresponds to transition, states undergoing transitions as projectors, and the parity picked up during transitions:}

%where $\Vec{b}\in S_{\hat{O}} = \{\Vec{b}|H(\Vec{b}) = N, \hat{O}\ket{\Vec{b}}\neq\ket{\Vec{b}}\}$, $\vec{a}$ induces state transitions via bitwise addition $\Vec{a}\oplus\Vec{b}$, and $\Vec{c}$ computes the parity information of $\Vec{b}$ via $\Vec{c}\cdot\Vec{b}$. They are directly connected to qubit operators:

\begin{equation}
    X^{\Vec{a}} = \prod^M_{m=1}\sigma_{x, m}^{\Vec{a}[m]}, 
    P^{\Vec{b}}=\prod_{m=1}^{M}\frac{1+(-1)^{\Vec{b}[m]}\sigma_{z, m}}{2},
    Z^{\Vec{c}} = \prod^M_{m=1}\sigma_{z, m}^{\Vec{c}[m]},
\end{equation}
dubbed as the X-string, P-string, and Z-string. Given an $M\times M$ encoder, such as the Bravyi-Kitaev (BK) encoding \cite{seeley2012bravyi}, the operators are encoded as \cite{Steudtner_2018}:
\begin{equation}
    \mathcal{E}\left(X^{\Vec{a}}\right) = X^{\mathbf{G}\Vec{a}},  
    \mathcal{E}\left(P^{\Vec{b}}\right) = P^{\mathbf{G}\Vec{b}},
    \mathcal{E}\left(Z^{\Vec{c}}\right) = Z^{(\mathbf{G}^{-1})^T\Vec{c}}.
    \label{eqn: Pauli-transformation}
\end{equation}

However, when $\mathbf{G}$ is $M\times Q$ with $M>Q$, the parity information cannot be represented as a Z-string operator due to the lack of a left inverse. Instead, Z-string representation is discarded via $Z^{\Vec{c}}P^{\Vec{b}} = (-1)^{\Vec{c}\cdot\Vec{b}}P^{\Vec{b}}$, leading to \cheng{instead} the following decomposition of $\hat{O}$:
%even this expression is wrong.
\begin{align}
    \Re/\Im{\hat{O}} &= X^{ \Vec{a}}\sum_{\Vec{b}\in S_{\hat{O}}}(-1)^{\Vec{c}\cdot\Vec{b}}(P^{ \Vec{b}}\pm P^{ (\Vec{a}\oplus\Vec{b})}),\label{eqn: unencoded re/im xp pair}\\
    \Re/\Im{\mathcal{E}(\hat{O})} &= X^{\mathbf{G}\Vec{a}}\sum_{\Vec{b}\in S_{\hat{O}}}(-1)^{\Vec{c}\cdot\Vec{b}}(P^{\mathbf{G}\Vec{b}}\pm P^{\mathbf{G}(\Vec{a}\oplus\Vec{b})}).\label{eqn: re/im xp pair}  
\end{align}
where we denote $S_{\hat{O}} = \{\forall\vec{b}\in \mathbb{F}^M_2|H(\vec{b}) =H(\vec{a}\oplus \vec{b})= N\}$ as the set of physical state basis where $\Vec{a}\oplus \Vec{b}$ is also a physical state. For a more thorough discussion of the mappings between classical bitstrings and quantum operators, see Appendix \ref{appendix: def}. 
%A concluding remark on this representation and decoding.
%we might need to remove randomized linear encoder.
%maybe we should demonstrate how the amplitudes can be computed. 

%\cheng{We observe that the symmetry constraint on $\mathbf{G}$ could be reformulated as the parity check matrix of its dual code $\mathbf{P}$. For any given $Q\times M$ linear compression $\mathbf{G}$, we can construct a $(M-Q)\times M$ dual code $\mathbf{P}^T$ such that they satisfy the kernel relationship $\mathbf{GP}^T = 0$. $\mathbf{P}$ can be interpreted as the generator of equivalent classes on the unencoded space, such that for every binary vectors $\Vec{b}\in \mathbb{F}^M_2$ and $\Vec{e}\in \mathbb{F}^{M-Q}_2$, $\Vec{b}\sim \Vec{b}\oplus \mathbf{P}^T \Vec{e}$. }
%\cheng{In this picture, the particle symmetry constraint $N$ requires that each Hamming weight $N$ unencoded codeword shall appear at most once in each equivalent class. It places a $2N+2$ lower bound on the Hamming weight the subcode of $\mathbf{P}^T$ which generates codewords with even weight. Thus, we can regard the even subcode of $\mathbf{P}^T$ as effectively a classical error correction code. Applying the Gilbert-Varshamov bound on the even subcode with dimension $(M-Q-1)\times M$ -- a lower bound for optimal rate, is equivalent to applying the upper bound on the parity check matrix $\mathbf{G}$ with dimension $Q\times M$. It leads to our first theorem:}

%I think you are finally getting there. 
%\textit{Fermionic encoding and decoding}\textemdash

\section{Qubit and Measurement Complexity}
\label{section: fermionic en/decoding}
%need to at least mention G once.
The connection between qubit operator and linear code allows us to incorporate classical \cheng{coding theory} to compute fermionic expectation values. Qubit operator encoding-decoding is translated to finding $\mathbf{G}$, which preserves $S_{\hat{O}}$, and the decoder $F: \mathbf{G}\Vec{b}\rightarrow \Vec{b}$. Parity check codes turns out to be the right candidate for this task. In our context, it distinctively encodes the conserved particle basis as fixed Hamming weight $\Vec{b}$ bit-flip errors and decodes them in polynomial time. CEC with distance $2N+1$ has a parity check matrix that distinctively encodes a number-conserved fermionic state. By applying the GV bound to the parity check code, we obtain an upper bound on the minimal size of the fermionic encoder, leading to our first result.

\begin{theorem}
The minimal qubit cost required to encode an $M$ fermionic modes $N$ electrons problem has an upper bound of $Q \leq 2N\log M$. \label{theorem: asymptotic scaling}
\end{theorem}
\begin{proof}
See Appendix \ref{appendix: proof}.
\end{proof}
This theorem asserts the existence of a logarithmic mode scaling linear encoder. Scalable construction of such encoders can be outsourced. Existing LDPC code packages such as AFF3CT \cite{cassagne2019aff3ct} and OpenFEC \cite{openfec} can generate not only scalable GV bound-converging encoders but also polynomial-time decoders. To study the impact of the optimal encoder on quantum chemistry problems, we introduce RLE \cheng{linear code} in Appendix \ref{appendix: RLE}. The compression rates of RLE $Q/M$ are presented in Figure \ref{Fig1} and Table \ref{t2} (Appendix \ref{appendix: numerical}).

%An optimal fermionic quantum state compression is dual to an optimal classical error correction code against the binary symmetric channel. Encoding which preserves the fermionic configurations is equivalent to the detection of distinctive noise in classical coding. 

Meanwhile, error decoding in CEC can be translated to decoding number-conserved states for quantum simulation. \cheng{However}, unlike classical decoding, the expectation values of $\hat{O}$ are decoded in multiple measurement bases, which contributes to the measurement complexity. To calculate the total measurement complexity, we note that $\Re/\Im{\hat{O}}$ in Eqn.~(\ref{eqn: re/im xp pair}) can be expanded as commuting Pauli strings measurable in one basis. It leads to the following measurement cost:
\begin{theorem}\label{thm: quantum chemistry ham scaling}
The number of Clifford bases for measuring the expectation value of a quantum chemistry Hamiltonian is upper bounded by $\mathcal{O}(M^4)$.
\end{theorem}
\begin{proof}
See Appendix \ref{appendix: measurement cost}.
\end{proof}
%Both theorems assure the existence of an encoder with optimal qubit-to-measurement complexity trade-off. 
\cheng{Upon sampling in the compressed measurement basis},
FED \cite{FED_code} is introduced in Algorithm \ref{alg2} to classical post-process the probability distributions of the projectors $X^{\mathbf{G}\Vec{a}}(P^{\mathbf{G}\Vec{b}}\pm P^{\mathbf{G}(\Vec{a}\oplus\Vec{b})})\rightarrow X^{ \Vec{a}}(P^{ \Vec{b}}\pm P^{ (\Vec{a}\oplus\Vec{b})})$, where the parity $(-1)^{\Vec{c}\cdot\Vec{b}}$ is subsequently computed.

\begin{algorithm}
    \caption{Fermionic Expectation Decoder}
    \label{alg2}
    \begin{algorithmic}
    \Require
      Probability distribution of projector variable $P(\Vec{d}\in \mathbb{F}^Q_2)$, $X$-string vector $\Vec{a}$ of $\hat{O}$, $Z$-string vector $\Vec{c}$ of $\hat{O}$, Parity Check Matrix $\mathbf{G}$, Minimal Weight Decoder $F(\cdot)$
    \Ensure
        Fermionic Expectation Value $\langle\hat{O}\rangle = \text{val}$
    \State
    $\Vec{e} \leftarrow \Vec{0}\in \mathbb{F}^Q_2$;
    \State
    $k \leftarrow \min{\{k, \Vec{a}[k] = 1\}}$;
    \State
    $\Vec{e}[k]\leftarrow 1$;
    \State
    Find $\mathbf{A}$ s.t. $\mathbf{A}\Vec{e} = \mathbf{G}\Vec{a}$;
    \State
    $\text{val} = 0$;
    \For{$P(\Vec{d}) \neq 0$}    
        \State
        $\Vec{m}\leftarrow \mathbf{A}\Vec{d}$;
        \State
        $\Vec{b}\leftarrow F(\Vec{m})$;
        \State
        Post-Select $\Vec{b}$ which obeys the number conservation constraint.
        \State
        $s\leftarrow (-1)^{\Vec{b}\cdot\Vec{c}}*(-1)^{\Vec{d}[k]}$;
        \State
        $\text{val}\leftarrow \text{val}+P(\Vec{d})*s$;
    \EndFor
    \State
    \Return $\text{val}$
    \end{algorithmic}
\end{algorithm}

\cheng{Compared to Ref. \cite{shee2022qubit}, FED avoids computing expectation value in the encoded subspace. It is compatible with all linear encoders, allowing resource leverages between the qubit and gate complexity while keeping the measurement efficient.} 
%this part needs more explanation.
%you will need to mention this optimal trade-off thing. You should cite the LDPC code generators in this paper as well.

%Overall, the duality between classical error correction and number-conserved quantum simulation not only allows us to scalably determines linear compression code with optimal qubit to Clifford measurement trade-off, it also resolves the scalability problem of measuring compressed fermionic observables. 

%This result is compatible with a number of simulation techniques such as the orbital optimization VQE simulation where both the qubit and measurement overheads are already small to begin with, and quantum signal processing \cite{}. 

%This result addresses the measurement problem appeared in Ref$\sim$\cite{shee2022qubit, Yen_2019} and first-quantized encoding \cite{lee2021variational}.

%I more or less understand how to reply to the gate complexity problem.

%\cheng{need to change the encoding.}

%In Figures \ref{fig: lih} and \ref{Fig3},

%\textit{Results}\textemdash
%\textit{Complexity Analysis} -- 
\section{Gate Complexity}
\label{section: gate complexity}
%I think this is a bad response.
% I need to think of how to reformulate the gate complexity section.

% part that we keep. 
The RLE encoder and its associated classical decoder have a computational complexity of $\mathcal{O}(M^N)$. In practice, they are replaced with polynomial encoder-decoder constructions that attain the GV bound, such as the Goppa code and Gallager's LDPC code. FED is fully scalable when combined with optimal CEC.

% language should be succinct, making it consistent with the original writing style. 

Such scalability is extended to the encoding of fermionic gate. In other words, the encoder-decoder formalism from error correction code also contributes directly to the gate complexity. We observe that the complexity growth is directly dual to the construction of QEC decoder, leading to the duality between QEC and QS. In what follows, we will formulate fermionic gate construction in the language of QEC decoder, and demonstrate its implementation in the language of decoders.

\subsection{QS-QEC Duality}
\begin{figure*}[htbp] % Use figure* for spanning both columns in revtex
    \centering
    % First row
    \subfloat[Circuit illustration of single round QEC.]{%
    $$\Qcircuit @C=1.5em @R=0.7em {
    \ket{QEC}&&\gate{\text{E}}&\ctrl{1}&\qw&\gate{\texttt{Decode}}&\qw&\qw&\qw \\
    \ket{0}&&\qw&\gate{\texttt{Encode}}&\meter&\control \cwx[-1] \cw&\ket{0}&&\qw
    }
    $$
    \label{fig:QEC}% 
    }
    \hspace{4em}
    \subfloat[General subroutine for implementing fermionic operators on linearly compressed fermionic state.]{%
    $$
    \Qcircuit @C=1.5em @R=0.7em {
    \ket{s}&&\gate{\texttt{Decode}}&\ctrl{1}            &\gate{\texttt{Decode}}&\qw&\ket{s}&&\qw \\
    \ket{QS}&&\ctrl{-1} &\gate{\texttt{Encode}}&\ctrl{-1} \qw&\qw&\qw&\qw&\qw
    }
    $$
    \label{fig:QS}%
    }
    \caption{QS-QEC duality. In QEC code, error channel on the $\ket{QEC}$ register is followed by error encoding and decoding after measurement, and the ancillary is reset to $\ket{0}$ state to be reused. The structure of error encoding and decoding finds its analogy in QS for gate operation, except that all operations are coherently implemented and two rounds are decoding is required to decouple $\ket{s}$, known as un-computation in quantum reversible logic. Uncomputation allows one to reset the ancillary qubits, contrasting the nature of QS and QEC.}
\end{figure*}

QEC involves error encoding via controlled gates, syndrome measurement, and error decoding, see Figure \ref{fig:QEC}.
In QS, the roles of the first and second registers are reversed. Number-conserved configurations are coherently stored as 'errors' in the second register. Rather than encoding errors, the encoder now maps quantum gate operations onto the parity-check qubits. Similarly, the decoder not only removes 'error' information from the first register but also decodes these configurations in the Jordan–Wigner (JW) basis. This duality leads to the fermionic gate operation picture shown in Figure \ref{fig:QS}.

In both cases, decoder complexity determines the scalability of the encodings. There are key differences between these protocols. In QS, the ancillary qubits $\ket{s}$ cannot always be the $\ket{0}$ state. Measurement and reset is forbidden to preserve the coherence, the decoder needs not to decode in the uncompressed number-conserved basis, and the decoder only needs to partially decode the necessary bit to encode gate operation. However, as we shall show, classical decoder from QEC (Figure \ref{fig:QEC}) can be directly quantized and mapped to (Figure \ref{fig:QS}), ultimately allowing us to characterize fermionic gate complexity with classical decoder complexities.

\subsection{Partial Decoding}

% This part is not completely right, and require revision. 
Although QEC protocol is designed to eliminate a complete set of quantum errors, such as single qubit bit-flip errors, syndrome measurements often reveal only \textit{local} errors that require local error corrections via partial classical decoders. Similarly, decoding local errors corresponds to local fermionic gate operations in QS. These processes do not require full knowledge of the number-conserved basis. The information needed depends on the particular gate one aims to implement. In fermionic simulation, the relevant operations are single and double excitation operators, which admit the following decomposition in the JW basis:
\begin{align}
    \hat{a}^\dagger_i\hat{a}_j\pm\hat{a}^\dagger_j\hat{a}_i &=
    X_iX_j\nonumber \\
    &\times \prod_{i<m<j}Z_m\nonumber \\
    &\times(P^0_iP^1_j\pm P^1_iP^0_j)   
    \label{eqn: single excitations}\\
\hat{a}^\dagger_i\hat{a}_j^\dagger\hat{a}_k\hat{a}_l\pm\hat{a}^\dagger_l\hat{a}^\dagger_k\hat{a}_j\hat{a}_i &= X_iX_jX_kX_l\nonumber \\
    &\times\prod_{i<m<j}Z_m\prod_{k<n<l}Z_n\nonumber \\
    &\times(P^0_iP^0_jP^1_kP^1_l\pm P^1_iP^1_jP^0_kP^0_l).
    \label{eqn: double excitations}
\end{align}
In the post-processing of FED, $X$-strings in the first row of Eqn.~(\ref{eqn: single excitations}) and (\ref{eqn: double excitations}) requires no decoding. Meanwhile, $Z$-strings in the second rows and the $P$ operators in the third rows requires decodings before expectation values are computed. Utilizing the QS-QEC duality, we recognize that logical gate operations shares the same complexity composition as FED -- decode and compute. The $P$ (third rows) operators are first decoded to select the fermionic modes ($i, j$ in Eqn.~(\ref{eqn: single excitations}) and $i, j, k, l$ in Eqn.~(\ref{eqn: double excitations})) for transitions induced by the $X$ and $Z$-strings. These decoded modes, denoted as $\vec{b}[i]\vec{b}[j]$ and $\vec{b}[i]\vec{b}[j]\vec{b}[k]\vec{b}[l]$, then computes fermionic transitions $\ket{\mathcal{E}(\vec{b})}\rightarrow \ket{\mathcal{E}(\vec{b}\oplus\vec{a})}$. Meanwhile, the decoded $Z$-strings (second rows), by summing over $\mathcal{O}(M)$ fermionic modes in the JW basis, compute the fermionic signs $\ket{\mathcal{E}(\vec{b})}\rightarrow \pm\ket{\mathcal{E}(\vec{b})}$.

%yes, this will be more succinct. But require extensive revision to the original plot. This seems to be the way to go. You do not need to include parity basis yes. 

% do not forget that in linear encoding one requires not trotterization.

For linearly encoded state, the summation and selection of $Z$ and $P$ operators requires distinctive decoder for fermionic mode $a$, and this necessitates $M$ single mode fermionic decoders to implement fermionic gate acting on arbitrary fermionic modes:
\[
\Qcircuit @C=0.5em @R=0.7em {
&\rho&&& \qw & \gate{D_{a}} & \qw &&&\times M\\
\ket{\mathcal{E}(\psi)}&&&& \qw & \ctrl{-1} & \qw &
}
\]
where $\rho$ is a single qubit density matrix and $a$ denotes the fermionic mode decoded from the linearly encoded multi-qubit state $\ket{\mathcal{E}(\psi)}$ on the second register. We note that these set of decoders are dual to QEC decoders where each error syndrome too corresponds to a unique decoding. Up next, we will demonstrate the constructions and the corresponding complexity of these single mode decoders, and subsequently, fermionic gate construction. 

\subsection{Decoders Construction}

% I am worried that they will ask for additional corrections. Possibly make a theorem in the beginning.
% you need to cite a well-known fact that 

Decoder from Figure \ref{fig:QEC} can be directly quantized to quantum decoder in \ref{fig:QS}, summarized as follow:
\begin{proposition}
    Classical linear decoder algorithms find direct quantum circuit representation.
\end{proposition}
\begin{proof}

By the virtue of Bennett's Theorem, any classical algorithm can be made reversible with at most polynomial overhead in space and time using ancillary bits \cite{Bennett1989}.

Next, due to deferred measurement principle \cite{Gurevich2021}, measurement probability distribution remains the same if one defers the measurement after a reversible classical logical gate operation, see Figure \ref{fig:quantization}. This property corresponds to an isomorphism between reversible classical logic to coherent quantum gate operations when processing unentangled classical bits. That is, for a general reversible classical logic gate $f: \mathbb{F}^A_2\rightarrow \mathbb{F}^A_2$, provided that it can be decomposed into universal reversible classical gate set, such as classical multi-control gates, it finds a direct quantum gate representation.

Since classical single mode decoders $\mathcal{D}_i(\cdot)$ may be decomposed as classical reversible logic (which is further decomposed into multi-control gates), for a classical single bit decoder:
\begin{equation}
    \mathcal{D}_i(\mathcal{E}(\vec{b})) = \mathcal{E}(\vec{b})\otimes\vec{b}[i]
\end{equation}
the same classical circuit $\hat{\mathcal{D}}_i(\cdot)$ can be directly implemented on quantum bits:
\begin{equation}
    \hat{\mathcal{D}}_i\left(\ket{\mathcal{E}(\vec{b})}\right) = \ket{\mathcal{E}(\vec{b})}\otimes\ket{\vec{b}[i]}
\end{equation}
It immediately follows that:
\begin{equation}
    \hat{\mathcal{D}}_i\left(\sum_{j}c_j\ket{\mathcal{E}(\vec{b}_j)}\right) = \sum_jc_j\ket{\mathcal{E}(\vec{b}_j)}\otimes\ket{\vec{b}_j[i]}
\end{equation}
due to linearity of quantum circuit.  
\end{proof}

Therefore, we have proved the existence of polynomial time quantum decoder for implementing single mode fermionic decoders. In fact, this result allows us to analyze any encoding (beyond linear) quantum gate complexity in terms of classical decoder complexity. The optimization of fermionic gate complexity boils down to the optimization of classical decoder complexity. This result unifies previous frameworks of second-quantized fermionic encodings, and provides a standardized mean to study qubit-gate-measurement complexity trade-offs in the language of decoders. 

\begin{figure}
\[
\Qcircuit @C=0.5em @R=0.7em {
&\qw&\ctrl{1}&\qw&\meter&&&\meter&\cctrl{1}&\cw&&\\
&\qw&\ctrl{1}&\qw&\meter&&&\meter&\cctrl{1}&\cw&&\\
&&&&&&&&&&&\\
&&\vdots&& &=& &&\vdots&&&\\
&&&&&&&&&&&\\
&&&&&&&&&&&\\
&\qw&\targ{0}\qwx[-1]&\qw&\meter&&&\qw&\targ{0}\cwx[-1]&\qw&\meter&\\
&\qw&\targ{0}\qwx[-1]&\qw&\meter&&&\qw&\targ{0}\cwx[-1]&\qw&\meter&\\
}
\]
    \caption{Deferred measurement principle for multi-control gate. It demonstrates direct correspondence between classical reversible circuit to quantum circuit.}
    \label{fig:quantization}
\end{figure}

\subsection{Decoding Subroutines}

Based on single mode fermionic decoders, we demonstrate subroutines that implement fermionic operators in the compressed subspace. We focus on projector $P$ decoding which enables selection of specific fermionic configurations required for the gate transitions, and Z-strings decoding which determines the sign structure in the JW basis. The latter can be implemented straightforwardly via the subroutine illustrated in Figure \ref{fig:sign-decoder}.
% designing this circuit is hard man. 
\begin{figure}
\[
\Qcircuit @C=0.5em @R=0.7em {
    & \qw & \gate{\hat{\mathcal{D}}_{\text{sign}(Z^{\vec{c}})}} & \qw & \quad = \quad &&&&
    & \ket{0} &&&& \gate{\hat{\mathcal{D}}_{a}} & \gate{\hat{\mathcal{D}}_{b}} & \qw && \dots &&& \gate{\hat{\mathcal{D}}_{z}} & \qw \\
    & \qw & \ctrl{-1} & \qw & &
    &&&& \ket{\mathcal{E}(\psi)} &&&& \ctrl{-1} & \ctrl{-1} & \qw && \dots &&& \ctrl{-1} & \qw \\
}
\]
\caption{Fermionic sign decoder circuit for the binary vector of the Z-string $Z^{\vec{c}}$ built from single mode fermionic decoders. The indices $a, b, ..., z$ denotes the non-0 entries of binary vectors $\vec{c}$.}
\label{fig:sign-decoder}
\end{figure}
The signs of mode from $a$ to $z$ are decoded and summed together coherently from the encoded state. After fermionic gate operation, the sign decoded fermionic modes remain the same, such that a second round decoding, as illustrated in Figure \ref{fig:QS}, un-compute the coherence on the ancillary qubit. 

Meanwhile, the decoder construction of the $P$ projector is less straight forward. In Eqn.~(\ref{eqn: single excitations}) and (\ref{eqn: double excitations}), decoded fermionic information of $P$ projectors is accompanied with transitions on fermionic mode $i,j$ and $i, j, k, l$ via $X_iX_j$ and $X_iX_jX_kX_l$, representing a pairs of states which we wish to decode to. These pairs of state undergoing transition must therefore be decoded to the same quantum state (indistinguishable to controlled state), and that the controlled state contains enough information to induce the right transition. This is achieved by preparing $\ket{s}$ as the Greenberger-Horne-Zeilinger (GHZ) state on the fermionic modes which required decoding. Suppose we have an pair of encoded fermionic states $\ket{\mathcal{E}(\vec{b})}, \ket{\mathcal{E}(\vec{b}\oplus\vec{a})}$ such that in our notation:
\begin{equation}
    X^{\vec{a}} = X_iX_j,
    \label{eqn: XX transition}
\end{equation}
acting on mode $i,j$, and a GHZ state on the same modes $\frac{\ket{0_i0_j}+\ket{1_i1_j}}{2}$. With decoder $\hat{\mathcal{D}}_i$ and $\hat{\mathcal{D}}_j$, if $\ket{\mathcal{E}(\vec{b})}$ is decoded to the value $m_in_j$, then $\ket{\mathcal{E}(\vec{b}+\vec{a})}$ is decoded to $(m\oplus1)_i(n\oplus 1)_j$ in the computational basis due to Eqn.~(\ref{eqn: XX transition}). However, when decoding the information on the GHZ state, we obtain the following:
\begin{align}
    m_in_j&\rightarrow\frac{\ket{(m\oplus 0)_i(n\oplus 0)_j}+\ket{(m\oplus 1)_i(n\oplus 1)_j}}{2}\\
    (m\oplus1)_i(n\oplus 1)_j&\rightarrow \frac{\ket{(m\oplus 1)_i(n\oplus 1)_j}+\ket{(m\oplus 0)_i(n\oplus 0)_j}}{2}.
\end{align}
That is, the resulting pair of GHZ states are indistinguishable despite of the outputs on the l.h.s. of the equation. Similar argument can be made with four modes decoding for double excitations. In that case we just prepare the state $\frac{\ket{0_i0_j0_k0_l}+\ket{1_i1_j1_k1_l}}{2}$ on the ancillary qubits. Figure \ref{cir:P-transition} illustrates double excitation $P$ decoding subroutine, where

\begin{figure}

$$
\Qcircuit @C=0.5em @R=0.7em {
\ket{0}&&\multigate{3}{GHZ}&\multigate{3}{\hat{\mathcal{D}}_{i, j, k, l}}&\ctrl{1}&\qw&\qw&\qw&\qw&\qw&\ctrl{1}&\multigate{3}{\hat{\mathcal{D}}_{i, j, k, l}}&\qw \\
\ket{0}&&\ghost{GHZ}&\ghost{\hat{\mathcal{D}}_{i, j, k, l}}&\targ{}&\ctrl{1}&\qw&\ctrl{1}&\qw&\ctrl{1}&\targ{}&\ghost{\hat{\mathcal{D}}_{i, j, k, l}}&\qw\\
\ket{0}&&\ghost{GHZ}&\ghost{\hat{\mathcal{D}}_{i, j, k, l}}&\qw&\targ{}&\ctrl{1}&\ctrlo{1}&\ctrl{1}&\targ{}&\qw&\ghost{\hat{\mathcal{D}}_{i, j, k, l}}&\qw\\
\ket{0}&&\ghost{GHZ}&\ghost{\hat{\mathcal{D}}_{i, j, k, l}}&\qw&\qw&\targ{}&\ctrl{1}&\targ{}&\qw&\qw&\ghost{\hat{\mathcal{D}}_{i, j, k, l}}&\qw\\
\ket{\mathcal{E}(\psi)}&&&\ctrl{-1}&\qw&\qw&\qw&\gate{e^{i\mathcal{E}(X^{\vec{a^*}})\theta}}&\qw&\qw&\qw&\ctrl{-1}&\qw %\gategroup{1}{3}{5}{4}{2em}{--}
}
$$
(a)

\[
\Qcircuit @C=0.5em @R=0.7em {
&\multigate{3}{\hat{\mathcal{D}}_{i, j, k, l}}&\qw&&&\gate{\hat{\mathcal{D}}_i}&\qw&\qw&\qw&\qw\\
&\ghost{\hat{\mathcal{D}}_{i, j, k, l}}&\qw&&&\qw&\gate{\hat{\mathcal{D}}_j}&\qw&\qw&\qw\\
&\ghost{\hat{\mathcal{D}}_{i, j, k, l}}&\qw&=&&\qw&\qw&\gate{\hat{\mathcal{D}}_k}&\qw&\qw\\
&\ghost{\hat{\mathcal{D}}_{i, j, k, l}}&\qw&&&\qw&\qw&\qw&\gate{\hat{\mathcal{D}}_l}&\qw\\
&\ctrl{-1}&\qw&&&\ctrl{-4}&\ctrl{-3}&\ctrl{-2}&\ctrl{-1}&\qw
}
\]
(b)
\caption{(a) Building block of fermionic double excitations built from single mode fermionic decoders. (b) Projector decoder based on Eqn.~(\ref{eqn: double excitations}).}
\label{cir:P-transition}
\end{figure}
\begin{equation}
    X^{\vec{a^*}} = X_iX_jX_kX_l.
\end{equation}
% introduce the linear program problem. You just need to propose one more. It is slightly difficult to motivate this result man. 
In this subroutine, $\ket{s} = \ket{GHZ}$, $\texttt{Decode} = \hat{\mathcal{D}}_{i,j,k,l}$ and circuit in between $\hat{\mathcal{D}}_{i,j,k,l}$ is the $\texttt{Encode}$ subroutine in Figure \ref{fig:QS}. The $\texttt{Encode}$ subroutine implements controlled rotation if the decoded state is $\frac{\ket{1_i1_j0_k0_l}+\ket{0_i0_j1_k1_l}}{2}$, corresponding to the $P$ projector in Eqn.~(\ref{eqn: double excitations}). The efficiency of one-to-one mapping of the X-strings (Eqn.~(\ref{eqn: Pauli-transformation})) is reflected the single controlled gate when implementing transition.

In Appendix \ref{appendix: gate complexity}, we will demonstrate time evolution encodings via $P$ and $Z$ decoding subroutine, and additionally discuss the worst case overhead when implementing gate complexity without fermionic decoders. 

% in this part, firstly, I do not want to make substantial changes, but changes are necessary.

%this paragraph requires modification.
%While our protocol has scalable encoding-decoding complexity, it does not directly address the encoded fermionic gate complexity. In fact, if we directly exponentiate the operator of Eqn.~(\ref{eqn: re/im xp pair}), the worst-case gate complexity scaling $\mathcal{O}(M^N)$ is recovered, as shown in Appendix \ref{appendix: measurement cost}. Unlike the segment code and Polylog scheme, the RLE encoder has no gate-related constraints. However, there are two ways to circumvent this problem.

%maybe put the computation in the main text?
\subsection{Near-Term Complexity}

Although QS-QEC duality demonstrates thoroughly complementary qubit-gate-measurement complexity contribution in the language of classical/quantum decoders, a complete fermionic simulation encoding only finds application in early FTQC regime. For near-term, our protocol can implement quantum algorithms without fermionic gates, as it is done in our VQE benchmarks. \cheng{Beside HEA ansatz, our protocol is compatible with the operator selection techniques in qubit ADAPT-VQE, which encodes the Hamiltonian information on the qubit operators rather than fermionic operators. In Appendix \ref{appendix: adapt-vqe}, we introduce a new algorithm which performs operator selection in the encoded subspace. The numerics illustrates the advantages of reducing operator actions in the compressed qubit space, suggesting potential benefits in minimizing the operator pools for implementing qubit ADAPT-VQE. Our results provide more effective ground state preparation strategies, expanding the scope of Variational Quantum Algorithms.}

%\cheng{Scalable gate implementation requires partial quantum state decoding $\mathcal{E}(\ket{\psi})\rightarrow\ket{\psi}$ where only the necessary fermionic modes are decoded. 
%Interestingly, encoding via parity check codes reveals a duality between quantum error correction (QEC) and quantum simulation (QS). In QEC, decoding parity-check qubits—corresponding to our compressed fermionic bits—is essential for error identification. In QS, this decoding translates to partial reconstruction of the fermionic state, enabling efficient gates in the Jordan-Wigner (JW) basis. This inspires a new approach to quantum gate design leveraging QEC techniques.} In Appendix \ref{appendix: gate complexity}, we outline a contruction of a coherent number-conserved decoder oracle. In practice, this QEC-inspired decoder is useful for designing sparse access oracles $\hat{O}_H$ and $\hat{O}_F$ \cite{sparse_access} within compressed JW bases, to extract matrix elements $\bra{\Vec{a}}\hat{H}\ket{\Vec{b}}$ and $\Vec{f}(\Vec{a}, l)$, where $l$ is the $l$th non-zero element in the matrix row $\Vec{a}$ in the Fock basis.

\section{Numerical Results}

\begin{figure}
    \centering
    \includegraphics[width=1.0\linewidth]{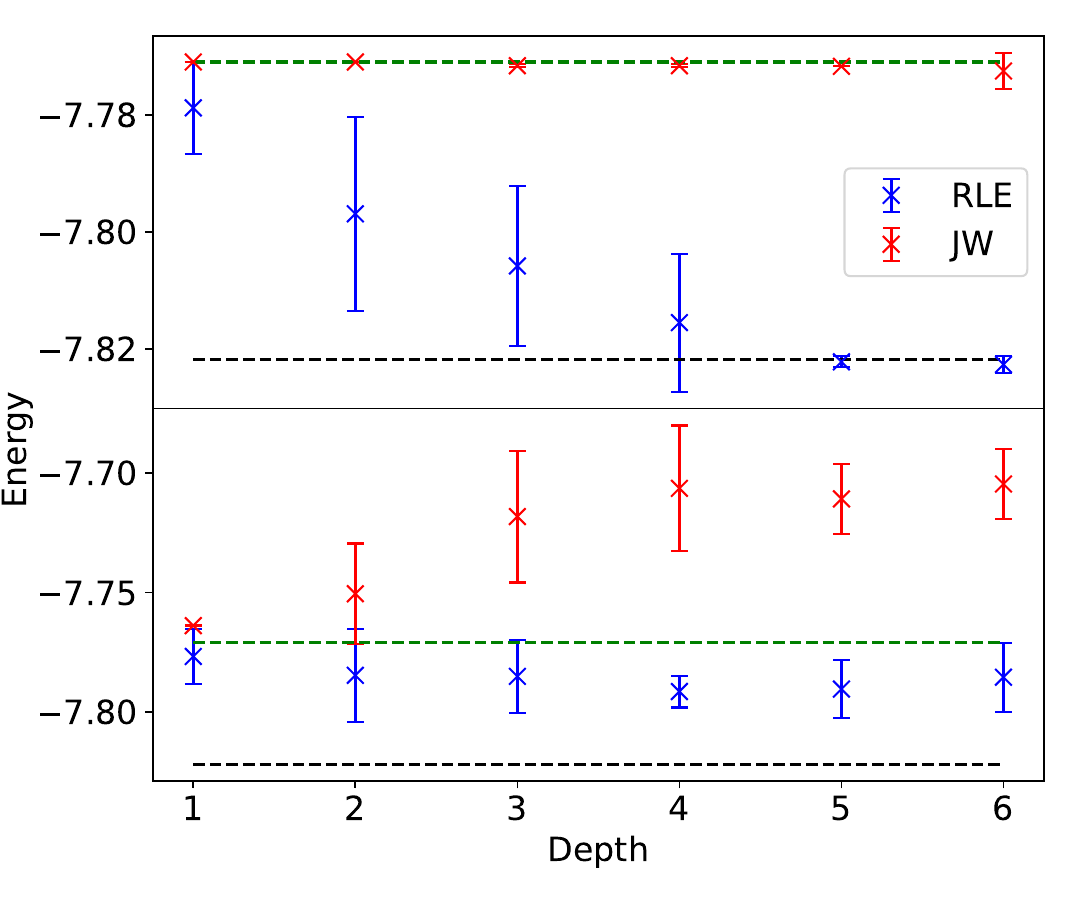}
    \caption{VQE energy of STO-3G LiH at 2.5 Å w/o compression using HEA at varying depth. The JW encoded noiseless energy (top red) cannot find energy better than the HF energy (green), whereas the RLE energy (blue) converges to chemical accuracy (black) at the 5th layer (25 CNOT). In the \texttt{ibm-sherbrooke} noisy simulation, the RLE energy shift is significantly less than the JW energy which diverges from the HF energy. }
    \label{fig: lih}
\end{figure}

\cheng{We benchmark the RLE-FED protocol on VQE trained with the L-BFGS-B optimizer, HEA ansatz and Hartree-Fock initialization. The impact of compression on gate resources and convergence w/o noise is investigated. In Figure \ref{fig: lih}, LiH molecule \cite{Rice_2021} in STO-3G compressed from $M=8, N = 2$ to $Q=6$ qubits is compared against the JW encoded LiH; in Figure \ref{Fig3}, the dissociation curve of $\text{H}_2$ molecule \cite{PhysRevLett.104.030502} in the 6-311G* basis compressed from $M = 12, N = 2$ to $Q=7$ qubits is simulated to studied its gate resource requirement. All data points are averaged over 30 instances of gradient descent to account for the randomness of our encoder and the quantum noise.}

%due to reduced CNOT noise and better representation of the number-conserved subspace.

Noiseless RLE simulated LiH energy converges to chemical accuracy ($<$1kcal/mol) with 30 local CNOT gates, with the noisy part converged to lower energy than the HF state. On the other hand, the JW encoded noiseless/noisy LiH simulation requires more entangling gates but fails to show any signs of convergence. Meanwhile, using as few as 66 local CNOT gates all $\text{H}_2$ energy points converge to chemical accuracy, compared to the JW encoded $\text{H}_2$ which requires 11 CNOT gates at least per layer, and double excitations which requires at least 13 CNOT gates. Training the JW encoded $\text{LiH}$ and $\text{H}_2$ molecules with HEA or the UCCSD ansatz often requires significantly more CNOT gates to achieve the same accuracy.
%HEA with the Unitary Coupled Cluster (UCC) ansatz \cite{Lee_2018, romero2018strategies, Shen_2017}, as shown
Subspace encoding allows state preparation with substantially less coherence time and consequently less circuit depth \cite{barren_plateau}. This is further supported with the benchmarks of 16-mode 6-31G $\text{LiH}$ in Tables \ref{t2} and \ref{t3} in Appendix \ref{appendix: numerical}.
\begin{figure}[!htb] 
    \centering
    \includegraphics[width=0.50\textwidth]{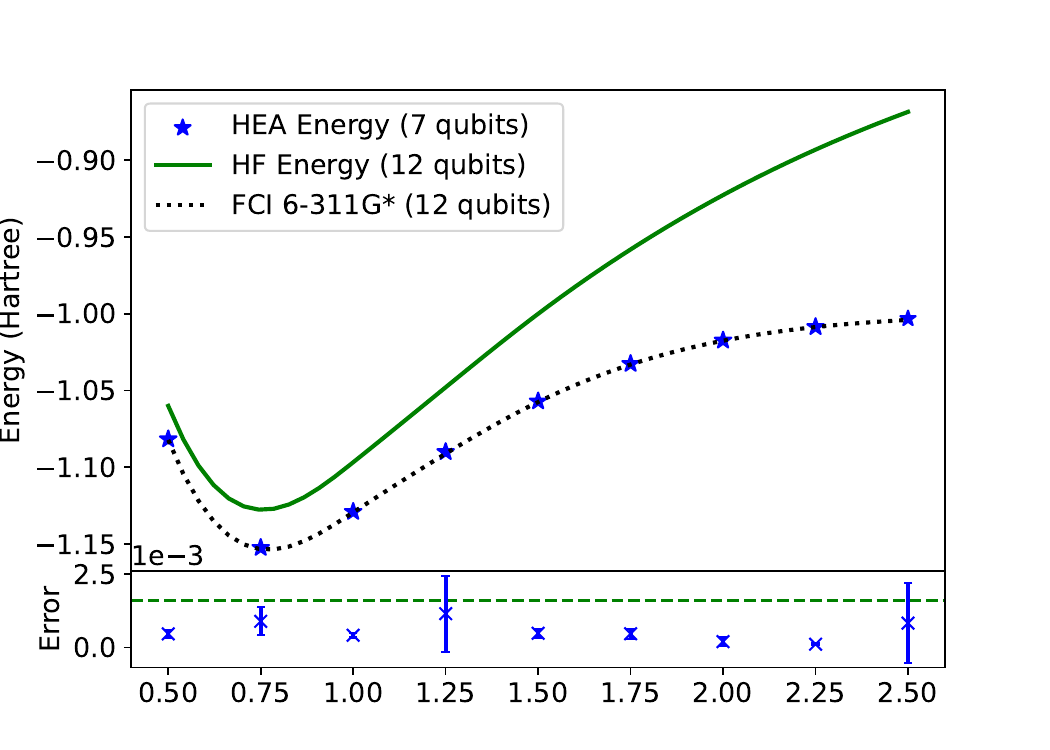}
    \caption{Noiseless VQE potential energy curve of 6-311G* (12 modes) $\text{H}_2$ potential energy curves encoded with 7 qubits using HEA with 11 layers (66 CNOT). The energy uncertainty is displayed at the bottom corresponds which lies within the chemical accuracy threshold (green dot).}
    \label{Fig3}
\end{figure}

%At near equilibrium bond length of 0.75 Angstrom, the HEA-RLE ground state energy has error $\Delta E = |E - E_{FCI}|=0.237$ Kcal/mol within chemical accuracy.

%depolarizing errors of $10^{-3}$ and $10^{-4}$ on CNOT and single-qubit gates respectively in Qiskit So, this is the set-up. You have the data so can write something.

Overall, successful ansatz training on the RLE encoded circuit demonstrates significantly fewer CNOT gates, and consequently, the spatial and temporal aspects of the quantum circuits are less prone to noise than the uncompressed case. Noisy energy variance shows that the compressed problem also produces a stable result. \cheng{While the compressed HEA simulation results does not directly imply scalability in VQE, our protocol is compatible to scalable VQE algorithm such as qubit ADAPT-VQE \cite{Qubit-Adapt-VQE}. See Appendix \ref{appendix: adapt-vqe} for implementation. Exponential improvement on the scaling of fermionic modes relieve the problem of barren plateau contributed from qubit scaling, and enhances the performance of qubit ADAPT-VQE. }

% need to emphasize the generality of FED, and the significance of the implications. 

%\textit{Conclusion} -- 
\section{Conclusion}

In this manuscript, we have \cheng{proposed the use of parity check matrix and FED for encoding and decoding number-conserved problems in quantum simulation}. We established a connection between number-conserved linear encoding and the parity check matrix of a dual CEC code, leading to new insights for constructing number-conserved encoding. Leveraging this duality, we showed that the optimal number-conserved linear encoding is bounded above by $2N\log M$. Our proposal also prepared efficient measurement bases for \cheng{number-conserved} observables, including the 2-RDM for quantum chemistry, with a measurement complexity upper bound of $1+\binom{M}{2}+\binom{M}{4}$. These results can be applied to construct encoding with optimal qubit-to-measurement trade-off. \cheng{Parity check code with polynomial time decoder are available in the Geometric Goppa and Gallager's LDPC codes. Meanwhile, FED offers a univeral protocol to decode fermionic information for all linear encodings in second-quantized simulation, enabling trade-offs between the qubit and gate complexity.}

To benchmark optimal codes, we constructed the RLE encoding and FED decoding algorithms. We compared VQE simulations of LiH in the RLE and JW encoded STO-3G/6-31G basis, and $\text{H}_2$ in the RLE encoded 6-311G* basis using HEA at different circuit depths. We conducted a VQE experiment on $\text{H}_2$ in the 6-311G* basis against full configuration interaction (FCI), demonstrating chemical accuracy across all bond lengths. All results favored the optimal encoding in terms of CNOT cost and convergence, which showcased the potential benefits of implementing VQE with our protocol. \cheng{We propose a new subroutine based on our protocol to implement qubit ADAPT-VQE in the encoded subspace, further reliving the resource requirement for variational quantum simulation of number-conserved quantum systems.}

% here, expand on the discussion.
\cheng{The duality between number conservation and error correction extends beyond FED. We show that QEC decoding is dual to simulating quantum gates in the compressed number-conserved subspace. We proposed a generalized framework on scalable fermionic gate encoding, echoing the results in Ref. \cite{qec-autoencoder}. This result opens the doors for designing sparse access oracles with technology in QEC decoders, expanding the current tool set of quantum simulation in FTQC regime and allowing flexible resource trade-offs when encoding number-conserved systems, while ensuring scalability.}

% it is getting there. 
Overall, the connection between CEC and number con-
servation has implications for resource trade-offs in quan-
tum simulation, both in the NISQ and FTQC regimes.
Our example demonstrates that this connection can be
leveraged to reduce quantum resource requirements for
variational quantum algorithms. We leave as an open question how these resource trade-offs can be further exploited in the encoding of more general quantum algorithms.

\textit{Acknowledgement} -- AH gratefully acknowledges the sponsorship from Research Grants Council of the Hong Kong Special Administrative Region, China (Project No. CityU 11200120), City University of Hong Kong (Project No. 7005615, 7006103), and CityU Seed Fund in Microelectronics (Project No. 9229135). MSK thanks the Samsung GRC project and the UK EPSRC (EP/@032643/1 and EP/Y004752/1). This work was supported by the project AnQuC-3 of the Competence Center Quantum Computing Rhineland Palatinate (Germany). All data and code supporting the findings of this study are openly available at Ref.~\cite{FED_code}.

% please add written out form of QEE, you never defined the abbreviation. Some more citations to go. What will they be?
%\bibliographystyle{apsrev4-1}
%unsrt 
\bibliography{linear.bib}

\begin{thebibliography}{43}%
\makeatletter
\providecommand \@ifxundefined [1]{%
 \@ifx{#1\undefined}
}%
\providecommand \@ifnum [1]{%
 \ifnum #1\expandafter \@firstoftwo
 \else \expandafter \@secondoftwo
 \fi
}%
\providecommand \@ifx [1]{%
 \ifx #1\expandafter \@firstoftwo
 \else \expandafter \@secondoftwo
 \fi
}%
\providecommand \natexlab [1]{#1}%
\providecommand \enquote  [1]{``#1''}%
\providecommand \bibnamefont  [1]{#1}%
\providecommand \bibfnamefont [1]{#1}%
\providecommand \citenamefont [1]{#1}%
\providecommand \href@noop [0]{\@secondoftwo}%
\providecommand \href [0]{\begingroup \@sanitize@url \@href}%
\providecommand \@href[1]{\@@startlink{#1}\@@href}%
\providecommand \@@href[1]{\endgroup#1\@@endlink}%
\providecommand \@sanitize@url [0]{\catcode `\\12\catcode `\$12\catcode `\&12\catcode `\#12\catcode `\^12\catcode `\_12\catcode `\%12\relax}%
\providecommand \@@startlink[1]{}%
\providecommand \@@endlink[0]{}%
\providecommand \url  [0]{\begingroup\@sanitize@url \@url }%
\providecommand \@url [1]{\endgroup\@href {#1}{\urlprefix }}%
\providecommand \urlprefix  [0]{URL }%
\providecommand \Eprint [0]{\href }%
\providecommand \doibase [0]{https://doi.org/}%
\providecommand \selectlanguage [0]{\@gobble}%
\providecommand \bibinfo  [0]{\@secondoftwo}%
\providecommand \bibfield  [0]{\@secondoftwo}%
\providecommand \translation [1]{[#1]}%
\providecommand \BibitemOpen [0]{}%
\providecommand \bibitemStop [0]{}%
\providecommand \bibitemNoStop [0]{.\EOS\space}%
\providecommand \EOS [0]{\spacefactor3000\relax}%
\providecommand \BibitemShut  [1]{\csname bibitem#1\endcsname}%
\let\auto@bib@innerbib\@empty
%</preamble>
\bibitem [{\citenamefont {Bauer}\ \emph {et~al.}(2020)\citenamefont {Bauer}, \citenamefont {Bravyi}, \citenamefont {Motta},\ and\ \citenamefont {Chan}}]{Bauer_2020}%
  \BibitemOpen
  \bibfield  {author} {\bibinfo {author} {\bibfnamefont {B.}~\bibnamefont {Bauer}}, \bibinfo {author} {\bibfnamefont {S.}~\bibnamefont {Bravyi}}, \bibinfo {author} {\bibfnamefont {M.}~\bibnamefont {Motta}},\ and\ \bibinfo {author} {\bibfnamefont {G.~K.-L.}\ \bibnamefont {Chan}},\ }\bibfield  {title} {\bibinfo {title} {Quantum algorithms for quantum chemistry and quantum materials science},\ }\href {https://doi.org/10.1021%2Facs.chemrev.9b00829} {\bibfield  {journal} {\bibinfo  {journal} {Chem. Rev.}\ }\textbf {\bibinfo {volume} {120}},\ \bibinfo {pages} {12685} (\bibinfo {year} {2020})}\BibitemShut {NoStop}%
\bibitem [{\citenamefont {Kitaev}\ and\ \citenamefont {Laumann}()}]{kitaev2009topological}%
  \BibitemOpen
  \bibfield  {author} {\bibinfo {author} {\bibfnamefont {A.}~\bibnamefont {Kitaev}}\ and\ \bibinfo {author} {\bibfnamefont {C.}~\bibnamefont {Laumann}},\ }\href@noop {} {\bibinfo {title} {Topological phases and quantum computation}},\ \Eprint {https://arxiv.org/abs/0904.2771} {arXiv:0904.2771} \BibitemShut {NoStop}%
\bibitem [{\citenamefont {Bauer}\ \emph {et~al.}(2023)\citenamefont {Bauer} \emph {et~al.}}]{bauer2022quantum}%
  \BibitemOpen
  \bibfield  {author} {\bibinfo {author} {\bibfnamefont {C.~W.}\ \bibnamefont {Bauer}} \emph {et~al.},\ }\bibfield  {title} {\bibinfo {title} {Quantum simulation for high-energy physics},\ }\href {https://doi.org/10.1103/PRXQuantum.4.027001} {\bibfield  {journal} {\bibinfo  {journal} {PRX Quantum}\ }\textbf {\bibinfo {volume} {4}},\ \bibinfo {pages} {027001} (\bibinfo {year} {2023})}\BibitemShut {NoStop}%
\bibitem [{\citenamefont {Ho}\ \emph {et~al.}(2018)\citenamefont {Ho}, \citenamefont {McClean},\ and\ \citenamefont {Ong}}]{ho_promise_2018}%
  \BibitemOpen
  \bibfield  {author} {\bibinfo {author} {\bibfnamefont {A.}~\bibnamefont {Ho}}, \bibinfo {author} {\bibfnamefont {J.}~\bibnamefont {McClean}},\ and\ \bibinfo {author} {\bibfnamefont {S.~P.}\ \bibnamefont {Ong}},\ }\bibfield  {title} {\bibinfo {title} {The {Promise} and {Challenges} of {Quantum} {Computing} for {Energy} {Storage}},\ }\href {https://doi.org/10.1016/j.joule.2018.04.021} {\bibfield  {journal} {\bibinfo  {journal} {Joule}\ }\textbf {\bibinfo {volume} {2}},\ \bibinfo {pages} {810} (\bibinfo {year} {2018})}\BibitemShut {NoStop}%
\bibitem [{\citenamefont {Santagati}\ \emph {et~al.}(2024)\citenamefont {Santagati}, \citenamefont {Aspuru-Guzik}, \citenamefont {Babbush}, \citenamefont {Degroote}, \citenamefont {González}, \citenamefont {Kyoseva}, \citenamefont {Moll}, \citenamefont {Oppel}, \citenamefont {Parrish}, \citenamefont {Rubin}, \citenamefont {Streif}, \citenamefont {Tautermann}, \citenamefont {Weiss}, \citenamefont {Wiebe},\ and\ \citenamefont {Utschig-Utschig}}]{santagati2023drug}%
  \BibitemOpen
  \bibfield  {author} {\bibinfo {author} {\bibfnamefont {R.}~\bibnamefont {Santagati}}, \bibinfo {author} {\bibfnamefont {A.}~\bibnamefont {Aspuru-Guzik}}, \bibinfo {author} {\bibfnamefont {R.}~\bibnamefont {Babbush}}, \bibinfo {author} {\bibfnamefont {M.}~\bibnamefont {Degroote}}, \bibinfo {author} {\bibfnamefont {L.}~\bibnamefont {González}}, \bibinfo {author} {\bibfnamefont {E.}~\bibnamefont {Kyoseva}}, \bibinfo {author} {\bibfnamefont {N.}~\bibnamefont {Moll}}, \bibinfo {author} {\bibfnamefont {M.}~\bibnamefont {Oppel}}, \bibinfo {author} {\bibfnamefont {R.~M.}\ \bibnamefont {Parrish}}, \bibinfo {author} {\bibfnamefont {N.~C.}\ \bibnamefont {Rubin}}, \bibinfo {author} {\bibfnamefont {M.}~\bibnamefont {Streif}}, \bibinfo {author} {\bibfnamefont {C.~S.}\ \bibnamefont {Tautermann}}, \bibinfo {author} {\bibfnamefont {H.}~\bibnamefont {Weiss}}, \bibinfo {author} {\bibfnamefont {N.}~\bibnamefont {Wiebe}},\ and\ \bibinfo {author} {\bibfnamefont {C.}~\bibnamefont {Utschig-Utschig}},\ }\bibfield  {title}
  {\bibinfo {title} {Drug design on quantum computers},\ }\href {https://doi.org/10.1038/s41567-024-02411-5} {\bibfield  {journal} {\bibinfo  {journal} {Nature Physics}\ }\textbf {\bibinfo {volume} {20}},\ \bibinfo {pages} {549} (\bibinfo {year} {2024})}\BibitemShut {NoStop}%
\bibitem [{\citenamefont {Zhang}\ \emph {et~al.}(2017)\citenamefont {Zhang}, \citenamefont {Hess}, \citenamefont {Kyprianidis}, \citenamefont {Becker}, \citenamefont {Lee}, \citenamefont {Smith}, \citenamefont {Pagano}, \citenamefont {Potirniche}, \citenamefont {Potter}, \citenamefont {Vishwanath}, \citenamefont {Yao},\ and\ \citenamefont {Monroe}}]{Zhang_2017}%
  \BibitemOpen
  \bibfield  {author} {\bibinfo {author} {\bibfnamefont {J.}~\bibnamefont {Zhang}}, \bibinfo {author} {\bibfnamefont {P.~W.}\ \bibnamefont {Hess}}, \bibinfo {author} {\bibfnamefont {A.}~\bibnamefont {Kyprianidis}}, \bibinfo {author} {\bibfnamefont {P.}~\bibnamefont {Becker}}, \bibinfo {author} {\bibfnamefont {A.}~\bibnamefont {Lee}}, \bibinfo {author} {\bibfnamefont {J.}~\bibnamefont {Smith}}, \bibinfo {author} {\bibfnamefont {G.}~\bibnamefont {Pagano}}, \bibinfo {author} {\bibfnamefont {I.-D.}\ \bibnamefont {Potirniche}}, \bibinfo {author} {\bibfnamefont {A.~C.}\ \bibnamefont {Potter}}, \bibinfo {author} {\bibfnamefont {A.}~\bibnamefont {Vishwanath}}, \bibinfo {author} {\bibfnamefont {N.~Y.}\ \bibnamefont {Yao}},\ and\ \bibinfo {author} {\bibfnamefont {C.}~\bibnamefont {Monroe}},\ }\bibfield  {title} {\bibinfo {title} {Observation of a discrete time crystal},\ }\href {https://doi.org/10.1038/nature21413} {\bibfield  {journal} {\bibinfo  {journal} {Nature}\ }\textbf {\bibinfo {volume} {543}},\ \bibinfo
  {pages} {217} (\bibinfo {year} {2017})}\BibitemShut {NoStop}%
\bibitem [{\citenamefont {Feynman}(1982)}]{feynman2018simulating}%
  \BibitemOpen
  \bibfield  {author} {\bibinfo {author} {\bibfnamefont {R.~P.}\ \bibnamefont {Feynman}},\ }\bibfield  {title} {\bibinfo {title} {Simulating physics with computers},\ }\href {https://doi.org/10.1007/BF02650179} {\bibfield  {journal} {\bibinfo  {journal} {Int. J. Theor. Phys}\ }\textbf {\bibinfo {volume} {21}} (\bibinfo {year} {1982})}\BibitemShut {NoStop}%
\bibitem [{\citenamefont {Jordan}\ and\ \citenamefont {Wigner}(1928)}]{Jordan1928wi}%
  \BibitemOpen
  \bibfield  {author} {\bibinfo {author} {\bibfnamefont {P.}~\bibnamefont {Jordan}}\ and\ \bibinfo {author} {\bibfnamefont {E.~P.}\ \bibnamefont {Wigner}},\ }\bibfield  {title} {\bibinfo {title} {{About the Pauli exclusion principle}},\ }\href {https://doi.org/10.1007/BF01331938} {\bibfield  {journal} {\bibinfo  {journal} {Z. Phys.}\ }\textbf {\bibinfo {volume} {47}},\ \bibinfo {pages} {631} (\bibinfo {year} {1928})}\BibitemShut {NoStop}%
\bibitem [{\citenamefont {Steudtner}\ and\ \citenamefont {Wehner}(2018)}]{Steudtner_2018}%
  \BibitemOpen
  \bibfield  {author} {\bibinfo {author} {\bibfnamefont {M.}~\bibnamefont {Steudtner}}\ and\ \bibinfo {author} {\bibfnamefont {S.}~\bibnamefont {Wehner}},\ }\bibfield  {title} {\bibinfo {title} {Fermion-to-qubit mappings with varying resource requirements for quantum simulation},\ }\href {https://doi.org/10.1088/1367-2630/aac54f} {\bibfield  {journal} {\bibinfo  {journal} {New J. Phys.}\ }\textbf {\bibinfo {volume} {20}},\ \bibinfo {pages} {063010} (\bibinfo {year} {2018})}\BibitemShut {NoStop}%
\bibitem [{\citenamefont {Bravyi}\ \emph {et~al.}()\citenamefont {Bravyi}, \citenamefont {Gambetta}, \citenamefont {Mezzacapo},\ and\ \citenamefont {Temme}}]{bravyi2017tapering}%
  \BibitemOpen
  \bibfield  {author} {\bibinfo {author} {\bibfnamefont {S.}~\bibnamefont {Bravyi}}, \bibinfo {author} {\bibfnamefont {J.~M.}\ \bibnamefont {Gambetta}}, \bibinfo {author} {\bibfnamefont {A.}~\bibnamefont {Mezzacapo}},\ and\ \bibinfo {author} {\bibfnamefont {K.}~\bibnamefont {Temme}},\ }\href@noop {} {\bibinfo {title} {Tapering off qubits to simulate fermionic hamiltonians}},\ \Eprint {https://arxiv.org/abs/1701.08213} {arXiv:1701.08213} \BibitemShut {NoStop}%
\bibitem [{\citenamefont {Kirby}\ \emph {et~al.}(2022)\citenamefont {Kirby}, \citenamefont {Fuller}, \citenamefont {Hadfield},\ and\ \citenamefont {Mezzacapo}}]{kirby2022second}%
  \BibitemOpen
  \bibfield  {author} {\bibinfo {author} {\bibfnamefont {W.}~\bibnamefont {Kirby}}, \bibinfo {author} {\bibfnamefont {B.}~\bibnamefont {Fuller}}, \bibinfo {author} {\bibfnamefont {C.}~\bibnamefont {Hadfield}},\ and\ \bibinfo {author} {\bibfnamefont {A.}~\bibnamefont {Mezzacapo}},\ }\bibfield  {title} {\bibinfo {title} {Second-quantized fermionic operators with polylogarithmic qubit and gate complexity},\ }\href {https://link.aps.org/doi/10.1103/PRXQuantum.3.020351} {\bibfield  {journal} {\bibinfo  {journal} {PRX Quantum}\ }\textbf {\bibinfo {volume} {3}},\ \bibinfo {pages} {020351} (\bibinfo {year} {2022})}\BibitemShut {NoStop}%
\bibitem [{\citenamefont {Shee}\ \emph {et~al.}(2022)\citenamefont {Shee}, \citenamefont {Tsai}, \citenamefont {Hong}, \citenamefont {Cheng},\ and\ \citenamefont {Goan}}]{shee2022qubit}%
  \BibitemOpen
  \bibfield  {author} {\bibinfo {author} {\bibfnamefont {Y.}~\bibnamefont {Shee}}, \bibinfo {author} {\bibfnamefont {P.~K.}\ \bibnamefont {Tsai}}, \bibinfo {author} {\bibfnamefont {C.~L.}\ \bibnamefont {Hong}}, \bibinfo {author} {\bibfnamefont {H.~C.}\ \bibnamefont {Cheng}},\ and\ \bibinfo {author} {\bibfnamefont {H.~S.}\ \bibnamefont {Goan}},\ }\bibfield  {title} {\bibinfo {title} {Qubit-efficient encoding scheme for quantum simulations of electronic structure},\ }\href {https://link.aps.org/doi/10.1103/PhysRevResearch.4.023154} {\bibfield  {journal} {\bibinfo  {journal} {Phys. Rev. Res.}\ }\textbf {\bibinfo {volume} {4}},\ \bibinfo {pages} {023154} (\bibinfo {year} {2022})}\BibitemShut {NoStop}%
\bibitem [{\citenamefont {Huggins}\ \emph {et~al.}(2021)\citenamefont {Huggins}, \citenamefont {McClean}, \citenamefont {Rubin}, \citenamefont {Jiang}, \citenamefont {Wiebe}, \citenamefont {Whaley},\ and\ \citenamefont {Babbush}}]{huggins_efficient_2021}%
  \BibitemOpen
  \bibfield  {author} {\bibinfo {author} {\bibfnamefont {W.~J.}\ \bibnamefont {Huggins}}, \bibinfo {author} {\bibfnamefont {J.~R.}\ \bibnamefont {McClean}}, \bibinfo {author} {\bibfnamefont {N.~C.}\ \bibnamefont {Rubin}}, \bibinfo {author} {\bibfnamefont {Z.}~\bibnamefont {Jiang}}, \bibinfo {author} {\bibfnamefont {N.}~\bibnamefont {Wiebe}}, \bibinfo {author} {\bibfnamefont {K.~B.}\ \bibnamefont {Whaley}},\ and\ \bibinfo {author} {\bibfnamefont {R.}~\bibnamefont {Babbush}},\ }\bibfield  {title} {\bibinfo {title} {Efficient and noise resilient measurements for quantum chemistry on near-term quantum computers},\ }\href {https://doi.org/10.1038/s41534-020-00341-7} {\bibfield  {journal} {\bibinfo  {journal} {Npj Quantum Inf.}\ }\textbf {\bibinfo {volume} {7}},\ \bibinfo {pages} {23} (\bibinfo {year} {2021})}\BibitemShut {NoStop}%
\bibitem [{\citenamefont {Seeley}\ \emph {et~al.}(2012)\citenamefont {Seeley}, \citenamefont {Richard},\ and\ \citenamefont {Love}}]{seeley2012bravyi}%
  \BibitemOpen
  \bibfield  {author} {\bibinfo {author} {\bibfnamefont {J.~T.}\ \bibnamefont {Seeley}}, \bibinfo {author} {\bibfnamefont {M.~J.}\ \bibnamefont {Richard}},\ and\ \bibinfo {author} {\bibfnamefont {P.~J.}\ \bibnamefont {Love}},\ }\bibfield  {title} {\bibinfo {title} {The bravyi-kitaev transformation for quantum computation of electronic structure},\ }\href {http://dx.doi.org/10.1063/1.4768229} {\bibfield  {journal} {\bibinfo  {journal} {J. Chem. Phys.}\ }\textbf {\bibinfo {volume} {137}},\ \bibinfo {pages} {224109} (\bibinfo {year} {2012})}\BibitemShut {NoStop}%
\bibitem [{\citenamefont {Shepherd}\ \emph {et~al.}(2012)\citenamefont {Shepherd}, \citenamefont {Grüneis}, \citenamefont {Booth}, \citenamefont {Kresse},\ and\ \citenamefont {Alavi}}]{Shepherd_2012}%
  \BibitemOpen
  \bibfield  {author} {\bibinfo {author} {\bibfnamefont {J.~J.}\ \bibnamefont {Shepherd}}, \bibinfo {author} {\bibfnamefont {A.}~\bibnamefont {Grüneis}}, \bibinfo {author} {\bibfnamefont {G.~H.}\ \bibnamefont {Booth}}, \bibinfo {author} {\bibfnamefont {G.}~\bibnamefont {Kresse}},\ and\ \bibinfo {author} {\bibfnamefont {A.}~\bibnamefont {Alavi}},\ }\bibfield  {title} {\bibinfo {title} {Convergence of many-body wave-function expansions using a plane-wave basis: From homogeneous electron gas to solid state systems},\ }\href {https://link.aps.org/doi/10.1103/PhysRevB.86.035111} {\bibfield  {journal} {\bibinfo  {journal} {Phys. Rev. B}\ }\textbf {\bibinfo {volume} {86}},\ \bibinfo {pages} {035111} (\bibinfo {year} {2012})}\BibitemShut {NoStop}%
\bibitem [{\citenamefont {Grüneis}\ \emph {et~al.}(2013)\citenamefont {Grüneis}, \citenamefont {Shepherd}, \citenamefont {Alavi}, \citenamefont {Tew},\ and\ \citenamefont {Booth}}]{Gruneis_2013}%
  \BibitemOpen
  \bibfield  {author} {\bibinfo {author} {\bibfnamefont {A.}~\bibnamefont {Grüneis}}, \bibinfo {author} {\bibfnamefont {J.~J.}\ \bibnamefont {Shepherd}}, \bibinfo {author} {\bibfnamefont {A.}~\bibnamefont {Alavi}}, \bibinfo {author} {\bibfnamefont {D.~P.}\ \bibnamefont {Tew}},\ and\ \bibinfo {author} {\bibfnamefont {G.~H.}\ \bibnamefont {Booth}},\ }\bibfield  {title} {\bibinfo {title} {{Explicitly correlated plane waves: Accelerating convergence in periodic wavefunction expansions}},\ }\href {https://doi.org/10.1063/1.4818753} {\bibfield  {journal} {\bibinfo  {journal} {J. Chem. Phys.}\ }\textbf {\bibinfo {volume} {139}},\ \bibinfo {pages} {084112} (\bibinfo {year} {2013})}\BibitemShut {NoStop}%
\bibitem [{\citenamefont {Bharti}\ \emph {et~al.}(2022)\citenamefont {Bharti}, \citenamefont {Cervera-Lierta}, \citenamefont {Kyaw}, \citenamefont {Haug}, \citenamefont {Alperin-Lea}, \citenamefont {Anand}, \citenamefont {Degroote}, \citenamefont {Heimonen}, \citenamefont {Kottmann}, \citenamefont {Menke}, \citenamefont {Mok}, \citenamefont {Sim}, \citenamefont {Kwek},\ and\ \citenamefont {Aspuru-Guzik}}]{Bharti_2022}%
  \BibitemOpen
  \bibfield  {author} {\bibinfo {author} {\bibfnamefont {K.}~\bibnamefont {Bharti}}, \bibinfo {author} {\bibfnamefont {A.}~\bibnamefont {Cervera-Lierta}}, \bibinfo {author} {\bibfnamefont {T.~H.}\ \bibnamefont {Kyaw}}, \bibinfo {author} {\bibfnamefont {T.}~\bibnamefont {Haug}}, \bibinfo {author} {\bibfnamefont {S.}~\bibnamefont {Alperin-Lea}}, \bibinfo {author} {\bibfnamefont {A.}~\bibnamefont {Anand}}, \bibinfo {author} {\bibfnamefont {M.}~\bibnamefont {Degroote}}, \bibinfo {author} {\bibfnamefont {H.}~\bibnamefont {Heimonen}}, \bibinfo {author} {\bibfnamefont {J.~S.}\ \bibnamefont {Kottmann}}, \bibinfo {author} {\bibfnamefont {T.}~\bibnamefont {Menke}}, \bibinfo {author} {\bibfnamefont {W.-K.}\ \bibnamefont {Mok}}, \bibinfo {author} {\bibfnamefont {S.}~\bibnamefont {Sim}}, \bibinfo {author} {\bibfnamefont {L.-C.}\ \bibnamefont {Kwek}},\ and\ \bibinfo {author} {\bibfnamefont {A.}~\bibnamefont {Aspuru-Guzik}},\ }\bibfield  {title} {\bibinfo {title} {Noisy intermediate-scale quantum algorithms},\ }\href
  {https://doi.org/10.1103%2Frevmodphys.94.015004} {\bibfield  {journal} {\bibinfo  {journal} {Rev. Mod. Phys.}\ }\textbf {\bibinfo {volume} {94}} (\bibinfo {year} {2022})}\BibitemShut {NoStop}%
\bibitem [{\citenamefont {Yen}\ \emph {et~al.}(2019)\citenamefont {Yen}, \citenamefont {Lang},\ and\ \citenamefont {Izmaylov}}]{Yen_2019}%
  \BibitemOpen
  \bibfield  {author} {\bibinfo {author} {\bibfnamefont {T.-C.}\ \bibnamefont {Yen}}, \bibinfo {author} {\bibfnamefont {R.~A.}\ \bibnamefont {Lang}},\ and\ \bibinfo {author} {\bibfnamefont {A.~F.}\ \bibnamefont {Izmaylov}},\ }\bibfield  {title} {\bibinfo {title} {Exact and approximate symmetry projectors for the electronic structure problem on a quantum computer},\ }\href {https://doi.org/10.1063%2F1.5110682} {\bibfield  {journal} {\bibinfo  {journal} {J. Chem. Phys.}\ }\textbf {\bibinfo {volume} {151}} (\bibinfo {year} {2019})}\BibitemShut {NoStop}%
\bibitem [{\citenamefont {Gilbert}(1952)}]{Gilbert_1952}%
  \BibitemOpen
  \bibfield  {author} {\bibinfo {author} {\bibfnamefont {E.~N.}\ \bibnamefont {Gilbert}},\ }\bibfield  {title} {\bibinfo {title} {A comparison of signalling alphabets},\ }\href {https://doi.org/10.1002/j.1538-7305.1952.tb01393.x} {\bibfield  {journal} {\bibinfo  {journal} {Bell Labs Tech. J}\ }\textbf {\bibinfo {volume} {31}},\ \bibinfo {pages} {504} (\bibinfo {year} {1952})}\BibitemShut {NoStop}%
\bibitem [{\citenamefont {Varshamov}(1957)}]{Varshamov_1957}%
  \BibitemOpen
  \bibfield  {author} {\bibinfo {author} {\bibfnamefont {R.~R.}\ \bibnamefont {Varshamov}},\ }\bibfield  {title} {\bibinfo {title} {Estimate of the number of signals in error correcting codes},\ }\href@noop {} {\bibfield  {journal} {\bibinfo  {journal} {Dokl. Akad. Nauk SSSR}\ }\textbf {\bibinfo {volume} {117}},\ \bibinfo {pages} {739–741} (\bibinfo {year} {1957})}\BibitemShut {NoStop}%
\bibitem [{\citenamefont {Katsman}\ \emph {et~al.}(1984)\citenamefont {Katsman}, \citenamefont {Tsfasman},\ and\ \citenamefont {Vladut}}]{Goppa_construction}%
  \BibitemOpen
  \bibfield  {author} {\bibinfo {author} {\bibfnamefont {G.}~\bibnamefont {Katsman}}, \bibinfo {author} {\bibfnamefont {M.}~\bibnamefont {Tsfasman}},\ and\ \bibinfo {author} {\bibfnamefont {S.}~\bibnamefont {Vladut}},\ }\bibfield  {title} {\bibinfo {title} {Modular curves and codes with a polynomial construction},\ }\href {https://doi.org/10.1109/TIT.1984.1056879} {\bibfield  {journal} {\bibinfo  {journal} {IEEE Trans. Inf. Theory}\ }\textbf {\bibinfo {volume} {30}},\ \bibinfo {pages} {353} (\bibinfo {year} {1984})}\BibitemShut {NoStop}%
\bibitem [{\citenamefont {Porter}\ \emph {et~al.}(1992)\citenamefont {Porter}, \citenamefont {Shen},\ and\ \citenamefont {Pellikaan}}]{Goppa_decoder}%
  \BibitemOpen
  \bibfield  {author} {\bibinfo {author} {\bibfnamefont {S.}~\bibnamefont {Porter}}, \bibinfo {author} {\bibfnamefont {B.-Z.}\ \bibnamefont {Shen}},\ and\ \bibinfo {author} {\bibfnamefont {R.}~\bibnamefont {Pellikaan}},\ }\bibfield  {title} {\bibinfo {title} {Decoding geometric goppa codes using an extra place},\ }\href {https://doi.org/10.1109/18.165441} {\bibfield  {journal} {\bibinfo  {journal} {IEEE Trans. Inf. Theory}\ }\textbf {\bibinfo {volume} {38}},\ \bibinfo {pages} {1663} (\bibinfo {year} {1992})}\BibitemShut {NoStop}%
\bibitem [{\citenamefont {Gallager}(1962)}]{Gallager}%
  \BibitemOpen
  \bibfield  {author} {\bibinfo {author} {\bibfnamefont {R.}~\bibnamefont {Gallager}},\ }\bibfield  {title} {\bibinfo {title} {Low-density parity-check codes},\ }\href {https://doi.org/10.1109/TIT.1962.1057683} {\bibfield  {journal} {\bibinfo  {journal} {IEEE Trans. Inf. Theory}\ }\textbf {\bibinfo {volume} {8}},\ \bibinfo {pages} {21} (\bibinfo {year} {1962})}\BibitemShut {NoStop}%
\bibitem [{\citenamefont {Mosheiff}\ \emph {et~al.}(2000)\citenamefont {Mosheiff}, \citenamefont {Resch}, \citenamefont {Ron-Zewi}, \citenamefont {Silas},\ and\ \citenamefont {Wootters}}]{List_Decoding}%
  \BibitemOpen
  \bibfield  {author} {\bibinfo {author} {\bibfnamefont {J.}~\bibnamefont {Mosheiff}}, \bibinfo {author} {\bibfnamefont {N.}~\bibnamefont {Resch}}, \bibinfo {author} {\bibfnamefont {N.}~\bibnamefont {Ron-Zewi}}, \bibinfo {author} {\bibfnamefont {S.}~\bibnamefont {Silas}},\ and\ \bibinfo {author} {\bibfnamefont {M.}~\bibnamefont {Wootters}},\ }\bibfield  {title} {\bibinfo {title} {Low-density parity-check codes achieve list-decoding capacity},\ }\href {https://doi.org/10.1137/20M1365934} {\bibfield  {journal} {\bibinfo  {journal} {SIAM Journal on Computing}\ }\textbf {\bibinfo {volume} {0}},\ \bibinfo {pages} {FOCS20} (\bibinfo {year} {2000})}\BibitemShut {NoStop}%
\bibitem [{\citenamefont {Tang}\ \emph {et~al.}(2021)\citenamefont {Tang}, \citenamefont {Shkolnikov}, \citenamefont {Barron}, \citenamefont {Grimsley}, \citenamefont {Mayhall}, \citenamefont {Barnes},\ and\ \citenamefont {Economou}}]{Qubit-Adapt-VQE}%
  \BibitemOpen
  \bibfield  {author} {\bibinfo {author} {\bibfnamefont {H.~L.}\ \bibnamefont {Tang}}, \bibinfo {author} {\bibfnamefont {V.}~\bibnamefont {Shkolnikov}}, \bibinfo {author} {\bibfnamefont {G.~S.}\ \bibnamefont {Barron}}, \bibinfo {author} {\bibfnamefont {H.~R.}\ \bibnamefont {Grimsley}}, \bibinfo {author} {\bibfnamefont {N.~J.}\ \bibnamefont {Mayhall}}, \bibinfo {author} {\bibfnamefont {E.}~\bibnamefont {Barnes}},\ and\ \bibinfo {author} {\bibfnamefont {S.~E.}\ \bibnamefont {Economou}},\ }\bibfield  {title} {\bibinfo {title} {Qubit-adapt-vqe: An adaptive algorithm for constructing hardware-efficient ansätze on a quantum processor},\ }\bibfield  {journal} {\bibinfo  {journal} {PRX Quantum}\ }\textbf {\bibinfo {volume} {2}},\ \href {https://doi.org/10.1103/prxquantum.2.020310} {10.1103/prxquantum.2.020310} (\bibinfo {year} {2021})\BibitemShut {NoStop}%
\bibitem [{\citenamefont {Low}\ and\ \citenamefont {Chuang}(2017)}]{Low_Cheung}%
  \BibitemOpen
  \bibfield  {author} {\bibinfo {author} {\bibfnamefont {G.~H.}\ \bibnamefont {Low}}\ and\ \bibinfo {author} {\bibfnamefont {I.~L.}\ \bibnamefont {Chuang}},\ }\bibfield  {title} {\bibinfo {title} {Optimal hamiltonian simulation by quantum signal processing},\ }\href {https://doi.org/10.1103/PhysRevLett.118.010501} {\bibfield  {journal} {\bibinfo  {journal} {Phys. Rev. Lett.}\ }\textbf {\bibinfo {volume} {118}},\ \bibinfo {pages} {010501} (\bibinfo {year} {2017})}\BibitemShut {NoStop}%
\bibitem [{\citenamefont {Kandala}\ \emph {et~al.}(2017)\citenamefont {Kandala}, \citenamefont {Mezzacapo}, \citenamefont {Temme}, \citenamefont {Takita}, \citenamefont {Brink}, \citenamefont {Chow},\ and\ \citenamefont {Gambetta}}]{kandala_hardware-efficient_2017}%
  \BibitemOpen
  \bibfield  {author} {\bibinfo {author} {\bibfnamefont {A.}~\bibnamefont {Kandala}}, \bibinfo {author} {\bibfnamefont {A.}~\bibnamefont {Mezzacapo}}, \bibinfo {author} {\bibfnamefont {K.}~\bibnamefont {Temme}}, \bibinfo {author} {\bibfnamefont {M.}~\bibnamefont {Takita}}, \bibinfo {author} {\bibfnamefont {M.}~\bibnamefont {Brink}}, \bibinfo {author} {\bibfnamefont {J.~M.}\ \bibnamefont {Chow}},\ and\ \bibinfo {author} {\bibfnamefont {J.~M.}\ \bibnamefont {Gambetta}},\ }\bibfield  {title} {{\selectlanguage {english}\bibinfo {title} {Hardware-efficient variational quantum eigensolver for small molecules and quantum magnets}},\ }\href {https://doi.org/10.1038/nature23879} {\bibfield  {journal} {\bibinfo  {journal} {Nature}\ }\textbf {\bibinfo {volume} {549}},\ \bibinfo {pages} {242} (\bibinfo {year} {2017})}\BibitemShut {NoStop}%
\bibitem [{\citenamefont {Zeng}\ \emph {et~al.}(2023)\citenamefont {Zeng}, \citenamefont {Fan}, \citenamefont {Liu}, \citenamefont {Li},\ and\ \citenamefont {Yang}}]{zeng2023quantum}%
  \BibitemOpen
  \bibfield  {author} {\bibinfo {author} {\bibfnamefont {X.}~\bibnamefont {Zeng}}, \bibinfo {author} {\bibfnamefont {Y.}~\bibnamefont {Fan}}, \bibinfo {author} {\bibfnamefont {J.}~\bibnamefont {Liu}}, \bibinfo {author} {\bibfnamefont {Z.}~\bibnamefont {Li}},\ and\ \bibinfo {author} {\bibfnamefont {J.}~\bibnamefont {Yang}},\ }\bibfield  {title} {\bibinfo {title} {Quantum {Neural} {Network} {Inspired} {Hardware} {Adaptable} {Ansatz} for {Efficient} {Quantum} {Simulation} of {Chemical} {Systems}},\ }\href {https://pubs.acs.org/doi/10.1021/acs.jctc.3c00527} {\bibfield  {journal} {\bibinfo  {journal} {J. Chem. Theory Comput.}\ }\textbf {\bibinfo {volume} {19}},\ \bibinfo {pages} {8587} (\bibinfo {year} {2023})}\BibitemShut {NoStop}%
\bibitem [{\citenamefont {Cassagne}\ \emph {et~al.}(2019)\citenamefont {Cassagne}, \citenamefont {Hartmann}, \citenamefont {Léonardon}, \citenamefont {He}, \citenamefont {Leroux}, \citenamefont {Tajan}, \citenamefont {Aumage}, \citenamefont {Barthou}, \citenamefont {Tonnellier}, \citenamefont {Pignoly}, \citenamefont {{Le Gal}},\ and\ \citenamefont {Jégo}}]{cassagne2019aff3ct}%
  \BibitemOpen
  \bibfield  {author} {\bibinfo {author} {\bibfnamefont {A.}~\bibnamefont {Cassagne}}, \bibinfo {author} {\bibfnamefont {O.}~\bibnamefont {Hartmann}}, \bibinfo {author} {\bibfnamefont {M.}~\bibnamefont {Léonardon}}, \bibinfo {author} {\bibfnamefont {K.}~\bibnamefont {He}}, \bibinfo {author} {\bibfnamefont {C.}~\bibnamefont {Leroux}}, \bibinfo {author} {\bibfnamefont {R.}~\bibnamefont {Tajan}}, \bibinfo {author} {\bibfnamefont {O.}~\bibnamefont {Aumage}}, \bibinfo {author} {\bibfnamefont {D.}~\bibnamefont {Barthou}}, \bibinfo {author} {\bibfnamefont {T.}~\bibnamefont {Tonnellier}}, \bibinfo {author} {\bibfnamefont {V.}~\bibnamefont {Pignoly}}, \bibinfo {author} {\bibfnamefont {B.}~\bibnamefont {{Le Gal}}},\ and\ \bibinfo {author} {\bibfnamefont {C.}~\bibnamefont {Jégo}},\ }\bibfield  {title} {\bibinfo {title} {Aff3ct: A fast forward error correction toolbox!},\ }\href@noop {} {\bibfield  {journal} {\bibinfo  {journal} {SoftwareX}\ }\textbf {\bibinfo {volume} {10}},\ \bibinfo {pages} {100432} (\bibinfo {year}
  {2019})}\BibitemShut {NoStop}%
\bibitem [{\citenamefont {McClean}\ \emph {et~al.}(2020)\citenamefont {McClean} \emph {et~al.}}]{openfec}%
  \BibitemOpen
  \bibfield  {author} {\bibinfo {author} {\bibfnamefont {J.~R.}\ \bibnamefont {McClean}} \emph {et~al.},\ }\bibfield  {title} {\bibinfo {title} {Openfermion: The electronic structure package for quantum computers},\ }\href {https://doi.org/10.1088/2058-9565/ab8ebc} {\bibfield  {journal} {\bibinfo  {journal} {Quantum Sci. Technol.}\ }\textbf {\bibinfo {volume} {5}},\ \bibinfo {pages} {034014} (\bibinfo {year} {2020})}\BibitemShut {NoStop}%
\bibitem [{\citenamefont {Cheng}\ and\ \citenamefont {Medina}(2024)}]{FED_code}%
  \BibitemOpen
  \bibfield  {author} {\bibinfo {author} {\bibfnamefont {M.~H.}\ \bibnamefont {Cheng}}\ and\ \bibinfo {author} {\bibfnamefont {A.~C.}\ \bibnamefont {Medina}},\ }\href@noop {} {\bibinfo {title} {Demonstration of fermionic expectation decoder}},\ \bibinfo {howpublished} {\url{https://github.com/Tcsg22/RLE-FED-protocol-demo}} (\bibinfo {year} {2024})\BibitemShut {NoStop}%
\bibitem [{\citenamefont {Bennett}(1989)}]{Bennett1989}%
  \BibitemOpen
  \bibfield  {author} {\bibinfo {author} {\bibfnamefont {C.~H.}\ \bibnamefont {Bennett}},\ }\bibfield  {title} {\bibinfo {title} {Time/space trade-offs for reversible computation},\ }\href {https://doi.org/10.1137/0219046} {\bibfield  {journal} {\bibinfo  {journal} {SIAM Journal on Computing}\ }\textbf {\bibinfo {volume} {18}},\ \bibinfo {pages} {766} (\bibinfo {year} {1989})}\BibitemShut {NoStop}%
\bibitem [{\citenamefont {Gurevich}\ and\ \citenamefont {Blass}(2022)}]{Gurevich2021}%
  \BibitemOpen
  \bibfield  {author} {\bibinfo {author} {\bibfnamefont {Y.}~\bibnamefont {Gurevich}}\ and\ \bibinfo {author} {\bibfnamefont {A.}~\bibnamefont {Blass}},\ }\bibfield  {title} {\bibinfo {title} {Quantum circuits with classical channels and the principle of deferred measurements},\ }\href {https://doi.org/https://doi.org/10.1016/j.tcs.2022.02.002} {\bibfield  {journal} {\bibinfo  {journal} {Theor. Comput. Sci}\ }\textbf {\bibinfo {volume} {920}},\ \bibinfo {pages} {21} (\bibinfo {year} {2022})}\BibitemShut {NoStop}%
\bibitem [{\citenamefont {Rice}\ \emph {et~al.}(2021)\citenamefont {Rice}, \citenamefont {Gujarati}, \citenamefont {Motta}, \citenamefont {Takeshita}, \citenamefont {Lee}, \citenamefont {Latone},\ and\ \citenamefont {Garcia}}]{Rice_2021}%
  \BibitemOpen
  \bibfield  {author} {\bibinfo {author} {\bibfnamefont {J.~E.}\ \bibnamefont {Rice}}, \bibinfo {author} {\bibfnamefont {T.~P.}\ \bibnamefont {Gujarati}}, \bibinfo {author} {\bibfnamefont {M.}~\bibnamefont {Motta}}, \bibinfo {author} {\bibfnamefont {T.~Y.}\ \bibnamefont {Takeshita}}, \bibinfo {author} {\bibfnamefont {E.}~\bibnamefont {Lee}}, \bibinfo {author} {\bibfnamefont {J.~A.}\ \bibnamefont {Latone}},\ and\ \bibinfo {author} {\bibfnamefont {J.~M.}\ \bibnamefont {Garcia}},\ }\bibfield  {title} {\bibinfo {title} {Quantum computation of dominant products in lithium{\textendash}sulfur batteries},\ }\href {https://doi.org/10.1063/5.0044068} {\bibfield  {journal} {\bibinfo  {journal} {J. Chem. Phys.}\ }\textbf {\bibinfo {volume} {154}},\ \bibinfo {pages} {134115} (\bibinfo {year} {2021})}\BibitemShut {NoStop}%
\bibitem [{\citenamefont {Du}\ \emph {et~al.}(2010)\citenamefont {Du}, \citenamefont {Xu}, \citenamefont {Peng}, \citenamefont {Wang}, \citenamefont {Wu},\ and\ \citenamefont {Lu}}]{PhysRevLett.104.030502}%
  \BibitemOpen
  \bibfield  {author} {\bibinfo {author} {\bibfnamefont {J.}~\bibnamefont {Du}}, \bibinfo {author} {\bibfnamefont {N.}~\bibnamefont {Xu}}, \bibinfo {author} {\bibfnamefont {X.}~\bibnamefont {Peng}}, \bibinfo {author} {\bibfnamefont {P.}~\bibnamefont {Wang}}, \bibinfo {author} {\bibfnamefont {S.}~\bibnamefont {Wu}},\ and\ \bibinfo {author} {\bibfnamefont {D.}~\bibnamefont {Lu}},\ }\bibfield  {title} {\bibinfo {title} {Nmr implementation of a molecular hydrogen quantum simulation with adiabatic state preparation},\ }\href {https://doi.org/10.1103/PhysRevLett.104.030502} {\bibfield  {journal} {\bibinfo  {journal} {Phys. Rev. Lett.}\ }\textbf {\bibinfo {volume} {104}},\ \bibinfo {pages} {030502} (\bibinfo {year} {2010})}\BibitemShut {NoStop}%
\bibitem [{\citenamefont {McClean}\ \emph {et~al.}(2018)\citenamefont {McClean}, \citenamefont {Boixo}, \citenamefont {Smelyanskiy}, \citenamefont {Babbush},\ and\ \citenamefont {Neven}}]{barren_plateau}%
  \BibitemOpen
  \bibfield  {author} {\bibinfo {author} {\bibfnamefont {J.~R.}\ \bibnamefont {McClean}}, \bibinfo {author} {\bibfnamefont {S.}~\bibnamefont {Boixo}}, \bibinfo {author} {\bibfnamefont {V.~N.}\ \bibnamefont {Smelyanskiy}}, \bibinfo {author} {\bibfnamefont {R.}~\bibnamefont {Babbush}},\ and\ \bibinfo {author} {\bibfnamefont {H.}~\bibnamefont {Neven}},\ }\bibfield  {title} {\bibinfo {title} {Barren plateaus in quantum neural network training landscapes},\ }\href {http://dx.doi.org/10.1038/s41467-018-07090-4} {\bibfield  {journal} {\bibinfo  {journal} {Nat. Commun.}\ }\textbf {\bibinfo {volume} {9}} (\bibinfo {year} {2018})}\BibitemShut {NoStop}%
\bibitem [{\citenamefont {Locher}\ \emph {et~al.}(2023)\citenamefont {Locher}, \citenamefont {Cardarelli},\ and\ \citenamefont {Müller}}]{qec-autoencoder}%
  \BibitemOpen
  \bibfield  {author} {\bibinfo {author} {\bibfnamefont {D.~F.}\ \bibnamefont {Locher}}, \bibinfo {author} {\bibfnamefont {L.}~\bibnamefont {Cardarelli}},\ and\ \bibinfo {author} {\bibfnamefont {M.}~\bibnamefont {Müller}},\ }\bibfield  {title} {\bibinfo {title} {Quantum error correction with quantum autoencoders},\ }\href {https://doi.org/10.22331/q-2023-03-09-942} {\bibfield  {journal} {\bibinfo  {journal} {Quantum}\ }\textbf {\bibinfo {volume} {7}},\ \bibinfo {pages} {942} (\bibinfo {year} {2023})}\BibitemShut {NoStop}%
\bibitem [{\citenamefont {Lee}\ \emph {et~al.}(2022)\citenamefont {Lee}, \citenamefont {Hsieh}, \citenamefont {Zhang},\ and\ \citenamefont {Shi}}]{lee2021variational}%
  \BibitemOpen
  \bibfield  {author} {\bibinfo {author} {\bibfnamefont {C.-K.}\ \bibnamefont {Lee}}, \bibinfo {author} {\bibfnamefont {C.-Y.}\ \bibnamefont {Hsieh}}, \bibinfo {author} {\bibfnamefont {S.}~\bibnamefont {Zhang}},\ and\ \bibinfo {author} {\bibfnamefont {L.}~\bibnamefont {Shi}},\ }\bibfield  {title} {\bibinfo {title} {Variational {Quantum} {Simulation} of {Chemical} {Dynamics} with {Quantum} {Computers}},\ }\href {https://doi.org/10.1021/acs.jctc.1c01176} {\bibfield  {journal} {\bibinfo  {journal} {J. Chem. Theory Comput.}\ }\textbf {\bibinfo {volume} {18}},\ \bibinfo {pages} {2105} (\bibinfo {year} {2022})},\ \bibinfo {note} {publisher: American Chemical Society}\BibitemShut {NoStop}%
\bibitem [{\citenamefont {MacWilliams}\ and\ \citenamefont {Sloane}(1977)}]{Intro_err}%
  \BibitemOpen
  \bibfield  {author} {\bibinfo {author} {\bibfnamefont {F.~J.}\ \bibnamefont {MacWilliams}}\ and\ \bibinfo {author} {\bibfnamefont {N.~J.~A.}\ \bibnamefont {Sloane}},\ }\href@noop {} {\emph {\bibinfo {title} {The Theory of Error Correcting Codes}}}\ (\bibinfo  {publisher} {North-Holland Pub. Co.},\ \bibinfo {address} {Amsterdam, The Netherlands},\ \bibinfo {year} {1977})\BibitemShut {NoStop}%
\bibitem [{\citenamefont {Babbush}\ \emph {et~al.}(2017)\citenamefont {Babbush}, \citenamefont {Berry}, \citenamefont {Sanders}, \citenamefont {Kivlichan}, \citenamefont {Scherer}, \citenamefont {Wei}, \citenamefont {Love},\ and\ \citenamefont {Aspuru-Guzik}}]{Babbush_2017}%
  \BibitemOpen
  \bibfield  {author} {\bibinfo {author} {\bibfnamefont {R.}~\bibnamefont {Babbush}}, \bibinfo {author} {\bibfnamefont {D.~W.}\ \bibnamefont {Berry}}, \bibinfo {author} {\bibfnamefont {Y.~R.}\ \bibnamefont {Sanders}}, \bibinfo {author} {\bibfnamefont {I.~D.}\ \bibnamefont {Kivlichan}}, \bibinfo {author} {\bibfnamefont {A.}~\bibnamefont {Scherer}}, \bibinfo {author} {\bibfnamefont {A.~Y.}\ \bibnamefont {Wei}}, \bibinfo {author} {\bibfnamefont {P.~J.}\ \bibnamefont {Love}},\ and\ \bibinfo {author} {\bibfnamefont {A.}~\bibnamefont {Aspuru-Guzik}},\ }\bibfield  {title} {\bibinfo {title} {Exponentially more precise quantum simulation of fermions in the configuration interaction representation},\ }\href {https://doi.org/10.1088/2058-9565/aa9463} {\bibfield  {journal} {\bibinfo  {journal} {Quantum Sci. Technol.}\ }\textbf {\bibinfo {volume} {3}},\ \bibinfo {pages} {015006} (\bibinfo {year} {2017})}\BibitemShut {NoStop}%
\bibitem [{\citenamefont {Sokolov}\ \emph {et~al.}(2020)\citenamefont {Sokolov}, \citenamefont {Barkoutsos}, \citenamefont {Ollitrault}, \citenamefont {Greenberg}, \citenamefont {Rice}, \citenamefont {Pistoia},\ and\ \citenamefont {Tavernelli}}]{sokolov2020quantum}%
  \BibitemOpen
  \bibfield  {author} {\bibinfo {author} {\bibfnamefont {I.~O.}\ \bibnamefont {Sokolov}}, \bibinfo {author} {\bibfnamefont {P.~K.}\ \bibnamefont {Barkoutsos}}, \bibinfo {author} {\bibfnamefont {P.~J.}\ \bibnamefont {Ollitrault}}, \bibinfo {author} {\bibfnamefont {D.}~\bibnamefont {Greenberg}}, \bibinfo {author} {\bibfnamefont {J.}~\bibnamefont {Rice}}, \bibinfo {author} {\bibfnamefont {M.}~\bibnamefont {Pistoia}},\ and\ \bibinfo {author} {\bibfnamefont {I.}~\bibnamefont {Tavernelli}},\ }\bibfield  {title} {\bibinfo {title} {Quantum orbital-optimized unitary coupled cluster methods in the strongly correlated regime: Can quantum algorithms outperform their classical equivalents?},\ }\href {http://dx.doi.org/10.1063/1.5141835} {\bibfield  {journal} {\bibinfo  {journal} {J. Chem. Phys.}\ }\textbf {\bibinfo {volume} {152}} (\bibinfo {year} {2020})}\BibitemShut {NoStop}%
\bibitem [{\citenamefont {Pless}(1998)}]{err_ch1}%
  \BibitemOpen
  \bibfield  {author} {\bibinfo {author} {\bibfnamefont {V.}~\bibnamefont {Pless}},\ }\bibinfo {title} {Introduction to the theory of error‐correcting codes}\ (\bibinfo  {publisher} {John Wiley /\& Sons, Ltd},\ \bibinfo {address} {Bognor Regis, U.K.},\ \bibinfo {year} {1998})\ Chap.\ \bibinfo {chapter} {1-2}, pp.\ \bibinfo {pages} {1--38}\BibitemShut {NoStop}%
\bibitem [{\citenamefont {Hernando}\ \emph {et~al.}(2019)\citenamefont {Hernando}, \citenamefont {Igual},\ and\ \citenamefont {Quintana-Ort\'{\i}}}]{Hernando_fast_BZ}%
  \BibitemOpen
  \bibfield  {author} {\bibinfo {author} {\bibfnamefont {F.}~\bibnamefont {Hernando}}, \bibinfo {author} {\bibfnamefont {F.~D.}\ \bibnamefont {Igual}},\ and\ \bibinfo {author} {\bibfnamefont {G.}~\bibnamefont {Quintana-Ort\'{\i}}},\ }\bibfield  {title} {\bibinfo {title} {Algorithm 994: Fast implementations of the brouwer-zimmermann algorithm for the computation of the minimum distance of a random linear code},\ }\href {https://dl.acm.org/doi/10.1145/3302389} {\bibfield  {journal} {\bibinfo  {journal} {ACM Trans. Math. Softw.}\ }\textbf {\bibinfo {volume} {45}} (\bibinfo {year} {2019})}\BibitemShut {NoStop}%
\end{thebibliography}%

\appendix

\section{Framework for number-conserved Linear Encoding }

% I would also like to complete this. Maybe do this at home?

%\cheng{put this somewhere in the appendix.}
%Unlike first-quantization, optimal linear compression on second-quantized problem avoids initial-state anti-symmetrization, and always achieves logarithmic scaling for all basis-set \cite{kirby2022second, Su_2021, babbush2017exponentially}. 

%\cheng{This part requires either a removal or rewrite.}
%Need to make the whole appendix looks good. First, you will need to discuss what to include here. 

%put abbreviation in the introduction. 

We provide a detailed discussion of the dual code correspondence that connects number-conserved linear encoding with classical error-correcting codes, inspired by Ref. \cite{bravyi2017tapering}.

In this work, we make use of a $Q \times M$ parity check code $\mathbf{G}$ of an $(M-Q) \times M$ CEC $\mathbf{C}$ to compress the number-conserved configurations. $\mathbf{G}$ is defined as the dual code of $\mathbf{C}$:
\begin{equation}
    \mathbf{G}\mathbf{C}^T = \mathbf{0}.
    \label{eqn: duality}
\end{equation}
Due to this duality, constraints on $\mathbf{P}$ can be directly translated to constraints in $\mathbf{C}$. With respect to $\mathbf{G}$, we can think of $\mathbf{C}$ as the generator of equivalence class $C$ on the $M$-dimensional binary vector space $\mathbb{F}^M_2$:
\begin{equation}
C = \{\vec{v} \in \mathbb{F}^M_2 \mid \forall \vec{c} \in \mathbb{F}^{M-Q}_2, \vec{v} = \mathbf{C}^T\vec{c}\} = \ker \mathbf{G}.
\end{equation}
This class $C$ is the kernel $\ker \mathbf{G}$ of the encoded space. $\ker \mathbf{G}$ also defines the equivalence class of all encoded bitstrings $\vec{b} \in \mathbb{F}^M_2$:
\begin{equation}
\vec{b} \sim \vec{b} \oplus \vec{v} \text{ for } \vec{v} \in C.
\end{equation}
Let us define $\vec{b} \oplus C$ to be the equivalence class of $\vec{b}$. To ensure that all configurations are distinctly encoded, number-conservation demands that each bitstring with Hamming weight $N$ shall only appear once in $\vec{b} \oplus C$, which constrains the Hamming weight of $C$ and thus its generator $\mathbf{C}$, the error-correcting code. If $\mathbf{C}$ generates a subspace of even Hamming weights with code distance $2N+2$, then for the subspace with Hamming weight less than or equal to $N$, $\mathbf{G}$ will distinctively encode them into the encoded space. Proof of this statement is in Lemma \ref{thm: criterion for successful encoding}. On the other hand, we can generate arbitrary odd equivalence classes because they do not contribute to mapping bitstrings with the same parity.

As we will show in the proof of Theorem \ref{theorem: asymptotic scaling}, this is equivalent to saying that we would like our $\mathbf{C}$ to have one unconstrained row element with odd Hamming weight. As for the rest of the rows, they form a code with even Hamming weight and distance $2N+2$. Thus, $\mathbf{C}$ contains an error-correcting subcode of dimension $(M-Q-1) \times M$ with even Hamming weight and distance $2N+2$ for us to construct a valid $Q \times M$ number-conserved linear encoding with code parameters $[M, M-Q-1, 2N+2]$.

\section{Framework for Qubit Operator Encoding}
\label{appendix: def}

The basics of operator encoding-decoding is reviewed here in the language of binary linear vector space. Literature such as Ref.~\cite{shee2022qubit} and \cite{lee2021variational} employ projectors for approaching the $\lceil N\log M\rceil$ lower bound. However, both methods suffer from unscalable measurement costs. One reason behind this broken scaling is the disruption of the linear structures of the projectors. In our framework, we connect the projector to linear binary vector space, which allows us to assess the measurement complexity with the linear structure.

%just cut as much as possible.
\subsection{Notations and Properties}
We define the notations of the $X$, $P$, $Z$-strings used in this paper.
\begin{definition}
    Projectors in the computational basis are defined as:
    \begin{equation}
    P^{0/1} = \ket{0/1}\bra{0/1}=\frac{1+(-1)^{0/1}Z}{2}
    \label{definition: proj}
    \end{equation}

\end{definition}
\begin{definition}
    Let $\Vec{b} \in \mathbb{F}_2^M$ be a binary vector, P-string is defined as follows:
    \begin{equation}
        P^{\Vec{a}}=\prod_{m=1}^{M}P_m^{\Vec{b}[m]}.
    \end{equation}
    with $\Vec{b}[m]$ denotes the $m$th entry of the vector and the following property:
    \begin{equation}
        P^{\Vec{b}}\ket{\Vec{a}} = \delta_{ab}\ket{\Vec{a}}.
        \label{eqn: P on bits}
    \end{equation}
    \label{definition: projector string}
\end{definition}
\begin{definition}\label{def: X-string}
    Let $\Vec{b} \in \mathbb{F}_2^M$ be a binary vector, X-string is defined as follows:
    \begin{equation}
        X^{\Vec{b}} = \prod^M_{m=1}X^{\Vec{b}[m]}_m.
    \end{equation}
    with $\Vec{b}[m]$ denotes the $m$th entry of the vector and the following property:
    \begin{equation}
        X^{\Vec{b}}\ket{a} = \ket{\Vec{a}\oplus \Vec{b}}.
        \label{eqn: X on bits}
    \end{equation}    
\end{definition}
\begin{definition}
Let $\Vec{b} \in \mathbb{F}_2^M$ be a binary vector, we define a Z-string as follows:
\begin{equation}
    Z^{\Vec{b}} = \prod^M_{m=1}Z^{\Vec{b}[m]}_m.
\end{equation} 
with $\Vec{b}[m]$ denotes the $m$th entry of the vector and the following property:
\begin{equation}
    Z^{\Vec{b}}\ket{\Vec{a}} = (-1)^{\Vec{b}\cdot \Vec{a}}\ket{\Vec{a}}.
    \label{eqn: Z on bits}
\end{equation}
\end{definition}
Next, we will derive the transformation of the $X$, $P$, and $Z$-strings with an encoder $\mathcal{E}: \ket{\vec{a}} \rightarrow \ket{\mathbf{G} \vec{a}}$. Immediately, we can infer from the state encoding of operators in Eqns~(\ref{eqn: P on bits}), (\ref{eqn: X on bits}), and (\ref{eqn: Z on bits}) the transformations as:
\begin{align}
    \mathcal{E}(P^{\Vec{b}}) &= P^{\mathbf{G}\Vec{b}}\label{eqn: P transform}\\
    \mathcal{E}(X^{\Vec{b}}) &= X^{\mathbf{G}\Vec{b}}\label{eqn: X transform}\\
    \mathcal{E}(Z^{\Vec{b}}) &= Z^{(\mathbf{G}^{-1})^T\Vec{b}}\label{eqn: Z transform}.
\end{align}
%we call it the parity check matrix here.
%maybe dont need to explain the Pauli Z string encoding.
%need to give review in the parity check code and error correction code.
Linear compression preserves Eqns~(\ref{eqn: P transform}) and (\ref{eqn: X transform}), but breaks Eqn.~(\ref{eqn: Z transform}) due to the lack of a left inverse for $\mathbf{G}$. Instead, the individual Pauli-Z operator can be encoded as:
\begin{equation}
    \mathcal{E}(Z^{\Vec{c}}) = \sum_{\Vec{b}\in S}(-1)^{\Vec{c}\cdot\Vec{b}}P^{\mathbf{G}\Vec{b}}
    \label{eqn: compressed Z operator}
\end{equation}
%a comprehensive change to the appendix. 
where $P^{\mathbf{G} \vec{b}}$ is a $Q$-qubit projector, and $S = \{\vec{b} \mid H(\vec{b}) = N\}$ is the set of all number-conserved states. The exact gate implementation of Eqn.~(\ref{eqn: compressed Z operator}) can be performed with multi-controlled gates or a linear combination of an exponential number of $Z$-strings. However, both methods suffer from scalability problems, as discussed in the gate complexity section in Appendix \ref{appendix: gate complexity}. This problem can be tackled using Quantum Signal Processing (QSP) techniques, provided that the encoded $\mathcal{E}(Z^{\vec{c}})$ is also sparse.

%what exactly do you want to say? Do you even know what you are talking about? Yeh man, otherwise you will just be wasting time. 
\subsection{XP Decomposition}

We can now derive a general representation of the qubit operator. In essence, all qubit operators consist of a selected set of states, state transitions, and phase evolution. In the computational basis, $X$-strings capture transitions, $Z$-strings capture phase, and the sum over projectors captures the states.

Given an X-string $X^{\Vec{a}}$ and a projector $P^{\Vec{b}}$, we can observe the following properties:
\begin{equation}
X^{\Vec{a}}P^{\Vec{b}}=P^{\Vec{a}\oplus\Vec{b}}X^{\Vec{a}}
\end{equation}
We can represent the Hermitian or anti-Hermitian part of the operator as:
\begin{align}
X^{\Vec{a}}(P^{\Vec{b}}\pm P^{\Vec{a}\oplus\Vec{b}}) &= \ket{\Vec{a}\oplus\Vec{b}}\bra{\Vec{b}}\pm\ket{\Vec{b}}\bra{\Vec{a}\oplus\Vec{b}}\\
(X^{\Vec{a}}(P^{\Vec{b}}\pm P^{\Vec{a}\oplus\Vec{b}}))^\dagger &= \pm X^{\Vec{a}}(P^{\Vec{b}}\pm P^{\Vec{a}\oplus\Vec{b}}).
\label{eqn: hermitian antihermitian}
\end{align}
Summation over this basis operator with respect to the Boolean vector $\Vec{a}$ and $\Vec{b}$ allows us to construct arbitrary (anti-)Hermitian qubit operators. As we show in Theorem \ref{thm: quantum chemistry ham scaling}, arbitrary sum over projector $P^{\Vec{b}}$ and its conjugate 
$P^{\Vec{a}\oplus\Vec{b}}$ and $X^{\Vec{a}}$ generates commuting Pauli strings as we expand the projector into $Z$-strings. The number of $X$-strings determines the measurement cost of a given qubit operator. Suppose now we exponentiate the Hermitian operator. It also determines the Trotterization steps required to approximate the unitary evolution.

In the JW representation, the product of creation and annihilation operators always generates one $X$-string under our decomposition. Our linear encoder then surjectively encodes the JW $X$-strings. This property allows us to exactly implement fermionic gates, though at the risk of an unscalable number of gates, and efficiently prepare measurement bases for fermionic observables.

\section{Qubit Complexity Proof}
\label{appendix: proof}
%\begingroup
%\def\thetheorem{\ref{thm: criterion for successful encoding}}
\begin{lemma}\label{thm: criterion for successful encoding}
    Given bitwise linear map $\mathcal{E}$ which is number-conserved, for all bitstrings $k$ which is an element of $\ker \mathcal{E}$, $D_H(k, 0)\geq 2N+2$ if $D_H(k, 0)\equiv 0 (\text{mod } 2)$.
\end{lemma}
%\addtocounter{theorem}{-1}
%\endgroup
\begin{proof}
This theorem claims that the even Hamming weight codewords, generated from the kernel space, must be lower bounded by $2N+2$. 

%this proof is shit and require revision.
Suppose there is a bitstring $\Vec{k}$ with Hamming weight $2K< 2N+2$. We can construct two distinct bitstrings with Hamming weight $N$ such that $\Vec{d}=\Vec{c}\oplus \Vec{k}$, and $\Vec{c}$, $\Vec{d}$ have $N-K$ overlapping bits. We observe the following relation in the encoded space
\begin{equation}
    \mathcal{E}(\Vec{d})=\mathcal{E}(\Vec{c})\oplus\mathcal{E}(\Vec{k})=\mathcal{E}(\Vec{c}).
\end{equation}
%so what is the consequence?
\end{proof}

\begingroup
\def\thetheorem{\ref{theorem: asymptotic scaling}}
\begin{theorem}
      Theorem \ref{theorem: asymptotic scaling}.  
\end{theorem}      
\addtocounter{theorem}{-1}
\endgroup
\begin{proof}
From Lemma \ref{thm: criterion for successful encoding}, we observe that all even kernel elements $k$, generated by the dual code $\mathbf{C}^T$, have even Hamming weight lower bounded as $D_H(k, 0) \geq 2N+2$. We can interpret the row of $\mathbf{C}$ as a $M-Q$ set of $M$-bit bitstrings, whose arbitrary bitwise addition $\oplus$ among the bitstrings are encoded into the $\mathbf{0}$ vector in the encoded space. Since $\mathbf{C}$ is invariant under elementary bitwise row addition, it is possible to eliminate the odd Hamming weight $M$-bit bitstrings with odd Hamming weight $M$-bit bitstrings. $\mathbf{C}$ can be reorganized as a set of $M$-bit bitstrings with \textit{only one} element with odd Hamming weight. In this form, we can impose Theorem \ref{thm: criterion for successful encoding} on the $M-Q-1$ numbers of $M$-bit bitstrings.

%you need to be a bit more elaborated. 
Next, we reformulate the constraint as finding the maximal $M-Q-1$ of $\mathbf{C}$ given qubits $M$, and minimal hamming distance $2N+2$, and all Hamming weight of elements in $\mathbf{C}$ are even. Borrowing from CEC, the maximum is upper bounded by the Hamming bound \cite{Intro_err} and lower bounded by the Gilbert–Varshamov bound \cite{Gilbert_1952, Varshamov_1957}. 
\begin{align}
    \frac{2^{M-1}}{\sum_{j=0}^{N}\binom{M}{2j}}&\leq 2^{M-Q-1}< \frac{2^{M-1}}{\sum_{j=0}^{\lfloor \frac{N}{2}\rfloor}\binom{M}{2j}}\label{eqn: even subspace bound}\\
    M-1-\log\sum_{j=0}^{N}\binom{M}{2j}&\leq M-Q-1 \nonumber\\
    &< M-1-\log\sum_{j=0}^{\lfloor N/2\rfloor}\binom{M}{2j} \nonumber\\
    2N\log M&\geq Q> 2\left\lfloor \frac{N}{2}\right\rfloor\log M.
\end{align}
We modify the Hamming and Gilbert-Varshamov bound in Eqn.~(\ref{eqn: even subspace bound}) such that we are only considering the bound on the even Hamming weight subspace, which has size $2^{M-1}$; and the sum $\sum_{j=0}^{N}\binom{M}{2j}$ is only over the even binomial coefficients. However, the argument of the Hamming cube as the maximal non-overlapping covering remained unchanged \cite{Gilbert_1952}. We also note that the Hamming bound for even subspace becomes a strict non-equality because the covering space does not cover the whole even code subspace. 

We have proved that optimal bitwise linear encoding has $\mathcal{O}(N\log M)$ scaling. In particular, the lower bound coincides with the scaling of physical states. 
\end{proof} 

\section{Proof for Measurement Complexity}
\label{appendix: measurement cost}

\begin{lemma}
    All fermionic operators that can be written as a product of creation and annihilation operators can be expressed in terms of the product of one X-string and sum over P-strings under linear compression.
    \label{lemma: XP representation}
\end{lemma}
\begin{proof}
%writing this thing is really not so easy. 
Let us consider the Jordan-Wigner basis. The one particle creation and annihilation operators are written as the product of the Pauli-Z operator $Z_m$, Pauli-X operator $X_i$, and projector $P^{0/1}_i$ defined in Definition \ref{definition: proj}:
\begin{align}
    \hat{a}^\dagger_i=\prod_{m=1}^{i-1}Z_m X_iP^0_i\\
    \hat{a}_j=\prod_{m=1}^{j-1}Z_m X_jP^1_j.
\end{align}
Using the completeness relation of the identity operator, it is possible to expand $P^{0/1}_i$ as a sum over projector strings $P^{a_1a_2...a_{M-1}}$ defined in Def. \ref{definition: projector string}:
\begin{equation}
    P^{0/1}_i = P^{0/1}_i\otimes \sum_{\Vec{a}\in \{0,1\}^{M-1}}P^{\Vec{a}}.
\end{equation}
Thus, the Z-strings, $\prod_{m=1}^{i-1}Z_m$ from the above example, can be absorbed into the projectors via the relationship:
\begin{equation}
    Z^{\Vec{c}}P^{\Vec{b}} = (-1)^{\Vec{c}\cdot\Vec{b}}P^{\Vec{b}}.
\end{equation}
After this procedure, fermionic observables can be represented with only the X-strings, projectors, and signs. Each projector represents a fermionic state, and the unphysical states are discarded such that the encoded projectors would not coincide with the number-conserved projectors due to the equivalent classes. Let us denote $\hat{O}$ as the product of creation and annihilation operators. Given a linear encoder $\mathcal{E}(\cdot)$ we can succinctly express the unencoded and encoded gates as follows:
\begin{align}
    \hat{O} &= X^{\Vec{a}}\sum_{\Vec{b}\in S_{\hat{O}}}(-1)^{\Vec{c}\cdot\Vec{b}}P^{\Vec{b}}\label{eqn: JW operator}\\
    \mathcal{E}(\hat{O}) &= X^{\mathbf{G}\Vec{a}}\sum_{\Vec{b}\in S_{\hat{O}}}(-1)^{\Vec{c}\cdot\Vec{b}}P^{\mathbf{G}\Vec{b}}\label{eqn: encoded operator}
\end{align}
where $S_{\hat{O}}$ represents the set of number-conserved states on which $\hat{O}$ acts, and $\Vec{a},\Vec{c}\in \mathbb{F}^M_2$ are binary vector uniquely associated with each operator $\hat{O}$, and $\Vec{a}$ induces physical state transition and $\Vec{c}$ encodes the sign information for each vector $\Vec{b}$.
\end{proof}

%\begingroup
%\def\thetheorem{\ref{theorem: one measurement basis}}
\begin{lemma}\label{lemma: one measurement basis}
    The real and imaginary part of an observable $\Re/\Im{\hat{O}}$ of the form in Eqn.~(\ref{eqn: encoded operator}) can be measured with one measurement basis under linear encoding.
\end{lemma}
%\addtocounter{theorem}{-1}
%\endgroup
\begin{proof}
    From Eqn.~(\ref{eqn: hermitian antihermitian}) and Eqn.~(\ref{eqn: encoded operator}), we can write the (anti-)Hermitian part of the operators as Eqn.~(\ref{eqn: re/im xp pair}):
    \begin{equation}
        \Re/\Im{\mathcal{E}(\hat{O})} = X^{\mathbf{G}\Vec{a}}\sum_{\Vec{b}\in S_{\hat{O}}}(-1)^{\Vec{c}\cdot\Vec{b}}(P^{\mathbf{G}\Vec{b}}\pm P^{\mathbf{G}(\Vec{a}\oplus\Vec{b})}).
        \label{eqn: appendix re/im xp pair}
    \end{equation}
    We make use of the following CNOT transformation to eliminate the extra Pauli-X operations:
    \begin{align}
    CNOT_{1, 2}X_1X_2CNOT_{1, 2}&=X_1\label{eqn: X and CNOT}\\
    CNOT_{1, 2}P^{a_1a_2}CNOT_{1, 2}&=P^{a_1(a_1\oplus a_2)}\label{eqn: P and CNOT}
    \end{align}
    where we define $i$ as control qubit and $j$ as target qubit for $CNOT_{i,j}$. Applying the Eqn.~(\ref{eqn: P and CNOT}) on a pair of conjugate projectors in the form of Eqn.~(\ref{eqn: hermitian antihermitian}), we observe that the projector at the target qubit will have the same parity. The full set of CNOT gates $\hat{C}$ results in the following transformation:
    \begin{align}
        \hat{C}X^{\mathbf{G}\Vec{a}}\hat{C}^\dagger &= X_B\label{eqn: CNOT X}\\
        \hat{C}P^{\mathbf{G}\Vec{b}}\hat{C}^\dagger &= P^{\mathbf{MG}\Vec{b}} = P^0_BP^{\Vec{m}}\label{eqn: CNOT 0}\\
        \hat{C}P^{\mathbf{G}(\Vec{a}\oplus\Vec{b})}\hat{C}^\dagger &= P^{\mathbf{MG}(\Vec{a}\oplus\Vec{b})}= P^1_BP^{\Vec{m}}\label{eqn: CNOT 1}
    \end{align}
    where the adjoint action of $\hat{C}$ on $P$ is equal to the $Q\times Q$ matrix $\mathbf{M}$ action on $\mathbf{G}\Vec{b}$, such that $\mathbf{MG}(n\Vec{a})\oplus\Vec{b} = ((n \mod 2)_B, \Vec{m}(\Vec{b}))$, with $B$ denoting the bit location on which $X_B$ acts and $n$ an integer. We write $\Vec{m} = \Vec{m}(\Vec{b})$ as a function of $\Vec{b}$ to simplify the notation in the sum. The operator becomes:
    \begin{align}
        \hat{C}\Re/\Im{\mathcal{E}(\hat{O})}\hat{C}^\dagger
        &= \sum_{\Vec{b}\in S_{\hat{O}}}(-1)^{\Vec{c}\cdot\Vec{b}}X_B(P_B^{0}P^{\Vec{m}(\Vec{b})}\pm P_B^{1}P^{\Vec{m}(\Vec{b})})\nonumber\\
        \label{eqn: transformed operator}
    \end{align}
    On qubit $B$, the sum of $P^0$ and $P^1$ generates local $I$ or difference generates $Z_N$. We can rewrite them as $X_N$ or $-iY_N$, as shown in Eqn.~(\ref{eqn: Gate decomposition}).

    Since all operators, except the qubit $Q$ which is either $X_i$ when hermitian or $iY_Q$ otherwise, are projectors made of $I, Z$ operators, the real/imaginary part of the bitwise linear encoded operators must generate commuting Pauli string in the Clifford transformed basis. Finally, since Clifford transforms one-to-one encoded Pauli strings, the operator $\Re/\Im{\mathcal{E}(\hat{O})}$ also generates commuting Pauli strings.
\end{proof}

\begingroup
\def\thetheorem{\ref{thm: quantum chemistry ham scaling}}
\begin{theorem}
    The number of Clifford bases for measuring the expectation value of a quantum chemistry Hamiltonian is upper bounded by $\mathcal{O}(M^4)$.
\end{theorem}      
\addtocounter{theorem}{-1}
\endgroup
\begin{proof}
    We only need to determine the measurement complexity of 2-RDM. 

    Next, from Lemma \ref{lemma: XP representation} and \ref{lemma: one measurement basis}, we see that the (anti-)Hermitian part of the 2-RDM, defined as $\bra{\psi}\hat{a}_i^\dagger\hat{a}_j^\dagger\hat{a}_k\hat{a}_l\pm\hat{a}_l^\dagger\hat{a}_k^\dagger\hat{a}_j\hat{a}_i\ket{\psi}$, can be measured in one basis for each set of fermionic modes $i, j, k, l$. In quantum chemistry, we are only interested in the real part of the observables. Therefore, each $i,j,k,l$ demands one measurement basis.

    Due to the result in Lemma \ref{lemma: one measurement basis}, each distinct X-string corresponds to one unique measurement basis. Upon encoding the 2-RDM operator via the JW transformation, we see that the encoded operator is of the form of the real part of Eqn.~(\ref{eqn: JW operator}). The X-string in this JW representation is invariant under any permutation among $i, j, k, l$ in the JW representation. Thus, the measurement scaling is equal to the possible combinations of $i, j, k, l$, where we allow the indices to be equal to each other. Next, we count the number of distinct X-strings derived from the possible permutation. In the JW representation, when all indices are not equal to each other, the number of distinct X-strings corresponds to a complete set of $M$ dimensional binary vectors with Hamming weight $4$. Meanwhile, when one pair of indices share the same digit, it corresponds to the set of $M$ dimensional binary vectors with Hamming weight $2$. When two pairs are equal, then it corresponds to $M$ dimensional binary vectors with Hamming weight $0$. Under our formalism, each X-string can be associated with a binary vector. Thus, the total measurement cost, without parallelization \cite{huggins_efficient_2021}, is equal to $\binom{M}{0}+\binom{M}{2}+\binom{M}{4}$. Parallelization is the simultaneous measurement of two or more sets of RDM elements with non-overlapping indices, maximal linear encoding breaks parallelization. 

     Finally, since a linear compression surjectively encodes the X-string guaranteed by the $2N+2$ code distance of the CEC $\mathbf{C}$, $\binom{M}{0}+\binom{M}{2}+\binom{M}{4}$ is the upper bound of the measurement scaling for the encoded space. 
\end{proof}

\begin{corollary}
The scaling of measurement basis can be completely characterized by the number of distinct X-strings that appears in the X and P decomposition.
\label{corollary: distinct X characterisation}
\end{corollary}
\begin{proof}
From the previous proof, each unique X-string, obtained from the X P decomposition of an operator, generates commuting Pauli strings. Meanwhile, Pauli strings generated by two distinct X-strings, and two pairs of respective projectors, contain mutually non-commuting terms.     
\end{proof}
This result can also be employed to characterize non-linear encoding, which \textit{does not} surjectively encode $X$ strings to $X$ strings. Instead, Eqn.~(\ref{eqn: appendix re/im xp pair}) will further split into a sum over $X$ strings. Scalable measurement is still possible if the distinct $X$ strings for each single/double excitation operators scales polynomially. 

%finally, we need to deal with gate complexity. Yeh man, 
\section{Fermionic Gate Complexity}
\label{appendix: gate complexity}

\subsection{Gate Decompositions without Decoder}

We examine fermionic gate complexity without quantum decoder. %Under the duality of \cheng{CEC and particle number conservation}, we show that linear encoding allows exact implementation of real or imaginary parts of the single and double excitation $\hat{O}$ in Eqn.~(\ref{eqn: encoded operator}) due to commutativity. Such property is missing for any non-linear encoding, where $\hat{O}$ is represented with more than one type of $X$-string under our $XP$ decomposition. Trotterisation is required to approximate $\hat{O}$ in non-linear encoding compared to linear encoding.

% you will be mainly cutting your stuff. 
%\cheng{However, commutativity alone} does not rescue linear encoders from unscalable gate complexity. \cheng{Nonetheless, it allows configurations to be simultaneously decoded, a feature that is dual to QEC error decoding. The correspondence provides a framework for implementing fermionic gates on compressed subspace using techniques inspired from quantum error decoder. This is analogous to the FED algorithm where fermionic observables are measured via classical error decoder. In this section, we will examine the implication of the QS-QEC duality on gate complexity, and the protocols for implementing gates on conserved subspace.}

\begin{proposition}
    The fermionic gate of single and double excitations can be decomposed as a series of multi-controlled gates and Clifford gates.
\end{proposition}
\begin{proof}
    From Lemma \ref{lemma: one measurement basis}, we can transform operators in Eqn.~(\ref{eqn: appendix re/im xp pair}) into X-string with the Clifford transformation $\hat{C}$ and operator of the form:
    \begin{equation}
        \hat{C}\Re/\Im{\mathcal{E}(\hat{O})}\hat{C}^\dagger = \sum_{\Vec{m}\in S_{\hat{O}}}X_B/-iY_BP^{\Vec{m}(\Vec{b})}
        \label{eqn: Gate decomposition}
    \end{equation}
    where $X_B / -iY_B$ is the last Pauli-$X/Y$ operator at qubit $B$. For the implementation of fermionic gates, it is necessary to describe the exponentiation of this operator representation. $\hat{C}$, described in Eqn.~(\ref{eqn: CNOT X}), (\ref{eqn: CNOT 0}), and (\ref{eqn: CNOT 1}), survives under exponentiation. Meanwhile, each operator $X_N / -iY_N P^{\Vec{m}(\Vec{b})}$ is mutually commutative, and its exponentiation corresponds to a multi-controlled $X/Y$ rotation, with the projectors acting as the control bits. Consequently, the exponentiation of the diagonalized fermionic operators can be represented as a product of multi-controlled Pauli-$X/Y$ gates.
\end{proof}
%we will need to address the gate complexity issue.

For molecular simulation, where only the real part of fermionic operators matters, the multi-controlled gates with a Clifford transformation can exactly implement fermionic gates defined in Lemma \ref{lemma: XP representation}. However, the number of multi-controlled gates given by the size of $S_{\hat{O}}$ indicates that the worst-case scenario gate complexity has combinatorial scaling $\mathcal{O}(M^N)$ \cite{bravyi2017tapering}. %\cheng{The non-scalability is analogous to the complexity challenge in FED where computation of expectation values on the compressed measurement basis is expensive. In other words, scalable gate implementation also requires fermionic information decoding to exploit the gate implementation advantages in second-quantized fermionic simulation.}

%The number of multi-controlled gates can hypothetically be reduced by gate approximations or employing the completeness relation of multi-controlled gates, where different controls can merge and cancel out each other based on their control state. However, classical parity check codes have no fermionic gate-related constraints on the encoder, making it difficult to expect any advantages in the scalable implementation of multi-controlled gates.

% this paragraph also requires improvement. 
The segment code in Ref. \cite{Steudtner_2018} and polylog code in Ref. \cite{kirby2022second} circumvent \cheng{the problem of decoding} by imposing constraints on the encoder. Segment code employs block structure on the encoder to ensure that the \cheng{compressed} $N$-particle states can be \cheng{locally distinguished upon encoding}, partially preserving locality in fermionic gate operations. Although the corresponding explicit gate complexity is not specified, it is polynomial in both $N$ and $M$ based on the construction. Meanwhile, the Polylog code further exploits \cheng{encoding with} block structure and Bravyi-Kitaev transformation for a more efficient fermionic information representation, and QSP to approximately implement single qubit Pauli Z operators, with a gate complexity lower scaling as $\mathcal{O}(N^2 \log M^4)$.

% this might be too much information. You got to cut down unnecessary information.
%The single Pauli-$Z$ operator is first encoded into a sum of Pauli-$Z$ operators, rescaled by a cosine function necessary for the final QSP polynomial transformation, which maps the cosine eigenvalues to a step function, the desired property of an encoded Pauli-$Z$ operator, we need to emphasize the fact that we are generalizing. This sentence. There are still repetition. 

These approaches demonstrate example of scalable fermionic gate implementation on the compressed state, but they are fundamentally limited to their encodings.  The QS-QEC duality provides a framework that generalizes these approaches. %In general, constraints on the encoder \cheng{ensures successful retrieval of} fermionic information such as sign and occupations from the encoded entangled state. %\cheng{It is under this background that we introduce the QS-QEC duality -- a unified framework for compressed fermionic gate operation via QEC protocol.}

% Narrating this is difficult, although not impossible. 
% Upon discussion 
% need to define QS

% you forgot to mention how many times the decoder require training. and also the fact that quantum algorithms.
% on top of that need to discuss about the complexity scaling. 

% just give comment here. You need to list out the key concepts and ingredients -- basis choice, reuseable qubits, 

% I think we will not talk about sparsity here. We need to just treat things as blackbox first in this discussion.

% again you will need to narrate a good transition here. 

% This part requires more structure. Be specific as well. 
% it is supposed to be short. 
\subsection{Time Evolution}

\subsubsection{Trottorization}
Building on Section \ref{section: gate complexity}, we demonstrate two examples of time evolution algorithm encoding based on subroutines introduced in Figure \ref{fig:sign-decoder} and \ref{cir:P-transition}. Firstly, we show that individual fermionic gates can be implemented with single mode decoders. Unitary evolution of the single and double excitation are implemented in Figure \ref{cir:single-excitation} and \ref{cir:double-excitation}.

For each query of the single and double excitation, the single mode decoders are queried at most $\mathcal{O}(M)$ times since the calculation of signs requires summing mode information from between fermionic modes $i<m<j$ for single excitations and additionally $k<n<l$ for double excitations. If the single and double excitations are queried $G$ times in Trottorization, the query complexity of the single mode decoders $D_i$ is given by:
\begin{equation}
    \mathcal{O}(GM).
\end{equation}

We demonstrate two examples of time evolution algorithm encoding based on the subroutines introduced in Figs. \ref{fig:sign-decoder} and \ref{cir:P-transition}. First, we show that individual fermionic gates can be implemented using single-mode decoders. The unitary evolution of single and double excitations is illustrated in Figs. \ref{cir:single-excitation} and \ref{cir:double-excitation}.

% you should just take your time and complete everything. 

While the discussion centers around the JW basis. In principle, decoding can be done in the BK basis, leading instead to the query complexity of:
\begin{equation}
    \mathcal{O}(G\log M),
\end{equation}
since only $\log M$ BK modes require decoding.

\subsubsection{Block Encoding}
Meanwhile, the decoder is not limited to encoding fermionic single and double excitations, it is possible to encode second-quantized subroutines, such as block encoding for time evolution \cite{Babbush_2017}. Given a second-quantized Hamiltonian of the form in Eqn. (\ref{eqn: hamiltonian}) with a Pauli decomposition:
\begin{equation}
    \hat{H} = \sum_{i}c_iP_i = \sum_{i}d_iZ^{\vec{c}_i}X^{\vec{a_i}},
\end{equation}
one can construct the $\mathcal{E}$\texttt{-SEL} (\texttt{SEL} stands for select) primitive in the compressed space using just the sign decoder in Figure \ref{fig:sign-decoder}:
\begin{align}
&\mathcal{E}\texttt{-SEL}\ket{\gamma}\ket{0}\ket{\mathcal{E}(\psi)}\nonumber\\
    &=\ket{\gamma}\hat{O}_{\text{sign }(\vec{c}_\gamma)}\ket{0}\mathcal{E}(X^{\vec{a}_{\gamma}})\ket{\mathcal{E}(\psi)}\nonumber\\
    &=\ket{\gamma}\ket{0}\mathcal{E}(Z^{\vec{c}_i})\mathcal{E}(X^{\vec{a}_{\gamma}})\ket{\mathcal{E}(\psi)}\nonumber\\
    &=\ket{\gamma}\ket{0}\mathcal{E}(P_{\gamma})\ket{\mathcal{E}(\psi)}
\end{align}
where the operator $\hat{O}_{\text{sign }(\vec{c}_i)}$ is given by the following sign decoder construction from Figure \ref{fig:sign-decoder}:
$$
\Qcircuit @C=0.5em @R=0.7em {
\ket{0}&&&\qw&\gate{\hat{\mathcal{D}}_{\text{sign}(\vec{c}_{\gamma})}}&\qw&\gate{Z}&\qw&\gate{\hat{\mathcal{D}}_{\text{sign}(\vec{c}_{\gamma})}}&\qw \\
\ket{\mathcal{E}(\psi)}&&&&\ctrl{-1}&\qw&\qw&\qw&\ctrl{-1}&\qw %\gategroup{1}{3}{5}{4}{2em}{--}
}
$$
When combined with the prepared primitive:
% dont be lazy in explanation
\begin{equation}
    \texttt{PREP}\ket{\gamma} = \sum_{i}\sqrt{\frac{{d_i}}{N}}\ket{i}
\end{equation}
we obtain the following block encoding:
\begin{align}
&\bra{0}\bra{0}\bra{\mathcal{E}(\psi)}\texttt{PREP}^\dagger \mathcal{E}\texttt{-SEL }\texttt{PREP}\ket{0}\ket{0}\ket{\mathcal{E}(\psi)}\nonumber \\ 
&= \bra{\mathcal{E}(\psi)}\mathcal{E}(\hat{H})\ket{\mathcal{E}(\psi)}/N
\end{align}
$\ket{\gamma}$ takes up $4\log M$ qubits, corresponding to the $\mathcal{O}(M^4)$ Pauli terms in $\hat{H}$. Thus, $G$ query of the $\mathcal{E}\texttt{-SEL}$ primitive, correponds to querying $\mathcal{O}(GM^4)$ $D_{\text{sign}(\vec{c})}$ decoder and hence:
\begin{equation}
    \mathcal{O}(GM^5)
\end{equation}
single mode fermionic decoders.

% first, you will need to get a second revision in this text, before moving on to the reply letter. 
% this sentence is again shit. 
%building this is insanely complex
\begin{figure*}[!th]
$$
\Qcircuit @C=0.5em @R=1em {
    \ket{0}&&\qw&\qw&\gate{\hat{\mathcal{D}}_{\text{sign}(Z^{\vec{c}})}}&\qw&\qw&\qw&\qw&\ctrl{2}&\ctrlo{2}&\qw&\qw&\qw&\qw&\gate{\hat{\mathcal{D}}_{\text{sign}(Z^{\vec{c}}))}}&\qw&  \\
    \ket{0}&&\qw&\multigate{1}{GHZ}&\qw&\multigate{1}{\hat{\mathcal{D}}_{i, j}}&\ctrl{1}&\qw&\qw&\qw&\qw&\qw&\qw&\ctrl{1}&\multigate{1}{\hat{\mathcal{D}}_{i, j}}&\qw&\qw&  \\
    \ket{0}&&\qw&\ghost{GHZ}&\qw&\ghost{\hat{\mathcal{D}}_{i, j}}&\targ{}&\qw&\qw&\ctrl{1}&\ctrl{1}&\qw&\qw&\targ{}&\ghost{\hat{\mathcal{D}}_{i, j}}&\qw&\qw&  \\
    \ket{\mathcal{E}(\psi)}&&&&\ctrl{-3}&\ctrl{-1}&\qw&\qw&\qw&\gate{R_{\mathcal{E}(X^{\vec{a}})}(\theta)}&\gate{R_{\mathcal{E}(X^{\vec{a}})}(-\theta)}&\qw&\qw&\qw&\ctrl{-1}&\ctrl{-3}&\qw&
}
$$
\caption{Single excitations quantum gate constructions based on sign and projector decoder in Figure \ref{fig:sign-decoder} and Figure \ref{cir:P-transition}. The $D$ operators represent the decoding subroutine and in between is the encoding subroutine in Figure \ref{fig:QS}. Definition of $\vec{a}$ and $\vec{c}$ follows the convention in Eqn.~(\ref{eqn: Pauli-transformation}).}
\label{cir:single-excitation}
\end{figure*}

\begin{figure*}[!th]
$$
\Qcircuit @C=0.5em @R=1em {
    \ket{0}&&\qw&\qw&\gate{\hat{\mathcal{D}}_{\text{sign}(Z^{\vec{c}^*})}}&\qw&\qw&\qw&\qw&\ctrl{2}&\ctrlo{2}&\qw&\qw&\qw&\qw&\gate{\hat{\mathcal{D}}_{\text{sign}(Z^{\vec{c}^*})}}&\qw&  \\
    \ket{0}&&\qw&\multigate{3}{GHZ}&\qw&\multigate{3}{\hat{\mathcal{D}}_{i, j, k, l}}&\ctrl{1}&\qw&\qw&\qw&\qw&\qw&\qw&\ctrl{1}&\multigate{3}{\hat{\mathcal{D}}_{i, j, k, l}}&\qw&\qw&  \\
    \ket{0}&&\qw&\ghost{GHZ}&\qw&\ghost{\hat{\mathcal{D}}_{i, j, k, l}}&\targ{}&\ctrl{1}&\qw&\ctrl{1}&\ctrl{1}&\qw&\ctrl{1}&\targ{}&\ghost{\hat{\mathcal{D}}_{i, j, k, l}}&\qw&\qw&  \\
    \ket{0}&&\qw&\ghost{GHZ}&\qw&\ghost{\hat{\mathcal{D}}_{i, j, k, l}}&\qw&\targ{}&\ctrl{1}&\ctrlo{1}&\ctrlo{1} &\ctrl{1}&\targ{}&\qw&\ghost{\hat{\mathcal{D}}_{i, j, k, l}}&\qw&\qw& \\
    \ket{0}&&\qw&\ghost{GHZ}&\qw&\ghost{\hat{\mathcal{D}}_{i, j, k, l}}&\qw&\qw&\targ{}&\ctrl{1}&\ctrl{1}&\targ{}&\qw&\qw&\ghost{\hat{\mathcal{D}}_{i, j, k, l}}&\qw&\qw&  \\
    \ket{\mathcal{E}(\psi)}&&&&\ctrl{-5}&\ctrl{-1}&\qw&\qw&\qw&\gate{R_{\mathcal{E}(X^{\vec{a}^*})}(\theta)}&\gate{R_{\mathcal{E}(X^{\vec{a}^*})}(-\theta)}&\qw&\qw&\qw&\ctrl{-1}&\ctrl{-5}&\qw&
}
$$
\caption{Double excitations quantum gate constructions based on sign and projector decoder in Figure \ref{fig:sign-decoder} and Figure \ref{cir:P-transition}. The $D$ operators represent the decoding subroutine and in between is the encoding subroutine in Figure \ref{fig:QS}. $\vec{a}^*$ and $\vec{c}^*$ denotes binary vectors for double excitations}
\label{cir:double-excitation}
\end{figure*}
% control the amount of input information. Need to compare to first quantized simulation too, and then. A detail discussion is required honestly. This is difficult man. The writing still requires much improvement.

%sorry, you forget to work hard. Now you will work hard again. 
%exponentially more precise paper. This might be too far? https://arxiv.org/pdf/1506.01020. 

% I think this might be too much, but lets write it out first.

% there is one more thing to discuss. Lets finish this one first. Will need to change to tongue of the appendix too. 

\subsection{Qubit-Gate complexity Trade-off}

Via the QS-QEC duality, we have demonstrated two examples of second-quantized subroutines in the linearly compressed basis using $\mathcal{O}(M)$ single-mode fermionic decoders. In exchange for the qubit reduction, the single-mode fermionic decoders result in an increase in gate complexity. The duality provides a natural framework for studying the qubit-to-gate complexity trade-off. In what follows, we will provide both lower and upper bounds on decoder complexity and give examples of efficient constructions.

Since the ancillary qubits are reusable, fermionic excitation evolution can be queried sequentially without increasing the number of ancillary qubits. In the absence of compression, the single-mode decoders are simply implemented using CNOT gates, as shown by the following relation: \begin{equation} \texttt{CNOT}_{i, a}\ket{m_i0_a} = \ket{m_im_a}, 
\label{eqn:CNOT-decoding}
\end{equation} where $m_i$ denotes the occupation number at mode $i$, and $a$ represents the ancillary qubit. Consequently, the $\hat{\mathcal{D}}_{i,j}/\hat{\mathcal{D}}_{i,j,k,l}$ decoders have $\mathcal{O}(1)$ complexity, while the $\hat{\mathcal{D}}_{\text{sign}}$ decoder has $\mathcal{O}(M)$ complexity.

In the worst-case scenario, the $\hat{\mathcal{D}}_{i,j}/\hat{\mathcal{D}}_{i,j,k,l}$ and $\hat{\mathcal{D}}_{\text{sign}}$ decoders can be trivially constructed using $\mathcal{O}(M^N)$ multi-controlled gates, similar to the construction in Eqn. (\ref{eqn: transformed operator}). Therefore, the single-mode fermionic decoder has complexity lower bound $C_{\hat{\mathcal{D}}_i}$ given by: \begin{equation} \mathcal{O}(1) < C_{\hat{\mathcal{D}}_i}. \end{equation} Using QS-QEC duality, the question of qubit-gate-measurement trade-off can be studied in terms of the decoder space-time complexity scaling as one compresses the qubit resources for the time resource.

\section{Fermionic Expectation Decoder}
\label{appendix: FED}

%I mean that part is still slightly unclear.
%you can state in a much more intuitive language of state projectors.
\subsection{Compatibility with Classical Decoder}
FED, as depicted in Algorithm \ref{alg2}, avoids direct computation of fermionic/spin operator in the encoded space through a decoding process. This approach circumvents the issue of measuring an unscalable number of Pauli operators and maintains polynomial scaling in measurement complexity. FED computes the expectation of the operator in Eqn.~(\ref{eqn: transformed operator}), from which a probability distribution $P(\vec{d})$ is sampled in the basis $\vec{d} = (1/0_B, \vec{m}(\vec{b}))$. To decode the bitstrings from $\mathbf{G}\vec{b}$ to $\vec{b}$, reversing the Clifford transformation is necessary to retrieve the projector $\mathbf{G}\vec{b}$ from $\vec{d}$—a task that is infeasible for classical post-processing.

Instead, when computing the fermionic sign $(-1)^{\vec{b}\cdot\vec{c}}$, one can focus solely on recovering either $P^{\mathbf{G}\vec{b}}$ or $P^{\mathbf{G}(\vec{a}\oplus\vec{b})}$ as shown in Eqn.~(\ref{eqn: re/im xp pair}) via partial recovery transformation. Collectively, this is equivalent to retrieving the projector information $X^{\mathbf{G}\vec{a}}(P^{\mathbf{G}\vec{b}}+P^{\mathbf{G}(\vec{a}\oplus\vec{b})})$. This can be achieved through the CNOT transformation $\hat{C}$ in Eqn.~(\ref{eqn: CNOT 0}) and (\ref{eqn: CNOT 1}), which can be efficiently implemented on a classical computer directly on the bitstrings of the histogram as $\mathbf{A}$.

Under our operator decomposition, we avoid two problems that arise in non-linear encoding: firstly, we avoid preparing an exponential number of measurement bases; and secondly, we avoid calculating expectation values in the encoded space, which would decompose into an exponential number of Pauli strings. The projector representation in Eqn.~(\ref{eqn: re/im xp pair}) avoids the Pauli representation of the encoded operator under conserved subspace compression. This new representation establishes a direct link to the decoding protocol of the CEC code and classifies the complexity of measurement grouping within the conserved subspace. Overall, this leads to a scalable encoding and decoding process for simulating systems with conserved particle numbers.
%need to mention how polynomial time decoding is possible over here. 
\subsection{The Decoding Algorithm}
\cheng{Results in Appendix \ref{appendix: measurement cost} can be directly employed for measurement basis preparations. Combined with the FED decoder, we have a quantum-classical observable decoder of compressed number-conserved fermionic data (Algorithm \ref{alg: decoder}):}

\begin{algorithm}
    \caption{DECODE}
    \label{alg: decoder}
    \begin{algorithmic}
    \Require
        Encoder $\mathcal{E}(\cdot)$, Decoder $F(\cdot)$, Encoded Quantum State $\ket{\mathcal{E}(\psi)}$, 
        Hamiltonian $\hat{H} = \sum_{i}c_{i}\hat{O}_i, \hat{O}_i^\dagger = \hat{O}_i$
    \Ensure
    $\bra{\psi}\hat{H}\ket{\psi}$
    \State
    $\text{sum} = 0$
    \For{$c_{i}, \hat{O}_i$ in $\hat{H}$}
    \State
    $\hat{C}\leftarrow\text{CLIFFORD}(\mathcal{E}(\hat{O}_i))$
    \State
    $\Vec{a}, \Vec{c}\leftarrow \text{GETVEC}(\hat{O}_i)$
    \State
    $P\leftarrow\text{MEAS}(\hat{C}^\dagger\ket{\mathcal{E}(\psi)})$
    \State
    $\text{val}\leftarrow\text{FED}(P, \Vec{a}, \Vec{c}, \mathcal{E}, F)$
    \State
    $\text{sum}\leftarrow \text{sum} + c_i*\text{val}$
    \EndFor
    \State
    \Return $\text{sum}$
    \end{algorithmic}
\end{algorithm}
\cheng{The Clifford gate preparation algorithm $\text{CLIFFORD}(\cdot)$ is based on the Clifford decomposition of Eqn.~(\ref{eqn: transformed operator}) and (\ref{eqn: Gate decomposition}) with a single qubit Clifford gate to diagonalize the $X_B/-iY_B$ Pauli operator. Meanwhile, $\text{GETVEC}(\cdot)$ is the extraction of the binary vector representation of $\hat{O}_i$ from the XP decomposition in the JW basis. $\text{MEAS}(\cdot)$ is measurement of the encoded quantum state in the quantum computer.}

\subsection{Applications}
% you need a final good summary.
The FED algorithm, leveraging the duality between CEC and fermionic observables, is applicable to general linearly compressed quantum states. Combining with the QS-QEC duality introduced in the Appendix \ref{appendix: gate complexity}, the FED algorithm finds applications across NISQ and FTQC regimes depending on the qubit-to-gate trade-offs in the compression, although the numerical results place no emphasis on the encoding of quantum gates.

For a moderate/Gilbert bounded qubit compression rate such as segment code and QNN decoders, although gate complexity increases, the FED algorithm finds application in decoding such compressed NISQ/early FTQC algorithms such as Trotterization. Meanwhile, as we will show in the following, the encoding/decoding protocol can be useful for encoding qubit ADAPT-VQE, providing a promising way of minimizing qubits for NISQ algorithms.

\section{qubit ADAPT-VQE}

\label{appendix: adapt-vqe}

% lets start writing.
In our numerical results, we have demonstrated the benefits of training VQE in a number-conserved subspace while ensuring that the encoding/decoding process remains scalable. However, the demonstration is limited to non-physically motivated ansätze, and it remains unexplored whether the RLE-FED protocol can facilitate the scalable implementation of VQE in the compressed subspace.

We present an algorithm that employs the RLE-FED scheme for encoding NISQ algorithms beyond non-physically motivated ansätze. While compression disrupts the structure of fermionic gates, it is possible to encode Hamiltonian information directly into qubit operators using qubit ADAPT-VQE \cite{Qubit-Adapt-VQE}. The algorithm selects the appropriate qubit gates by identifying the qubit operator that maximizes the gradient within an operator pool:
    \begin{equation}
        \frac{\partial}{\partial \theta_i}\langle\hat{H}\rangle = \bra{\mathcal{E}(\psi)}[\mathcal{E}(\hat{H}), \hat{\tau_i}]\ket{\mathcal{E}(\psi)},
        \label{eqn: gradient}
    \end{equation}
\cheng{
where $\hat{H}$ is the fermionic Hamiltonian, $\hat{\tau_i}$ is a qubit operator belonged to the operator pool, and $\ket{\mathcal{E}(\psi)}$ represents the encoded fermionic state in the qubit space. Specifically, to implement the operator $[\mathcal{E}(\hat{H}), \hat{\tau_i}]$, we make use of the parameter shift rule. $\mathcal{E}(\hat{H})$ is the Hamiltonian in the encoded space which requires decoding, and $\hat{\tau}_i$ is a qubit operator that acts directly on the qubit space. One can act $R_{\hat{\tau}_i}(\pm\pi/2)$ directly on $\ket{\mathcal{E}(\psi)}$ and compute the gradient via decoding the following term:}

\begin{equation}
    \frac{\partial}{\partial \theta_i}\langle\hat{H}\rangle = \nonumber
\end{equation}
\begin{equation}
    \frac{\langle R_{\hat{\tau}_i}(\pi/2)^\dagger\mathcal{E}(\hat{H})R_{\hat{\tau}_i}(\pi/2)\rangle -\langle R_{\hat{\tau}_i}(-\pi/2)^\dagger\mathcal{E}(\hat{H})R_{\hat{\tau}_i}(-\pi/2)\rangle}{2}.
\end{equation}
\cheng{In particular, the terms $\langle R_{\hat{\tau}_i}(\pi/2)^\dagger\mathcal{E}(\hat{H})R_{\hat{\tau}_i}(\pi/2)\rangle$ and $\langle R_{\hat{\tau}_i}(-\pi/2)^\dagger\mathcal{E}(\hat{H})R_{\hat{\tau}_i}(-\pi/2)\rangle$ can be classically decoded. This leads to the subroutine for qubit operator selection in qubit compressed ADAPT-VQE in Algorithm \ref{alg4}. Qubit ADAPT-VQE shows evidence that avoids barren-plateau and local minima, and it demonstrates CNOT gate reduction compared to the implementing fermionic ADAPT-VQE.}

\begin{algorithm}
    \caption{Qubit Operator Optimizer}
    \label{alg4}
    \begin{algorithmic}
    \Require
        Encoder $\mathcal{E}(\cdot)$, Decoder $F(\cdot)$, Qubit operator pool $\text{Pool}=\{\hat{\tau}_1, \hat{\tau}_2, ..., \hat{\tau}_n\}$, Hamiltonian $\hat{H}$, compressed quantum state $\ket{\mathcal{E}(\psi)}$
    \Ensure
        $\text{Max}_{\hat{\tau}_i\in P}\bra{\mathcal{E}(\psi)}[\mathcal{E}(\hat{H}), \hat{\tau_i}]\ket{\mathcal{E}(\psi)}$
    \State
    $L = [0]$
    \State
    $g_{\text{max}}=0$
    \For{$\hat{\tau}_m \in \text{Pool}$}
    \State
    $\ket{\psi^+} = e^{i\hat{\tau}_m\pi/2}\ket{\mathcal{E}(\psi)}$
    \State
    $\ket{\psi^-} = e^{-i\hat{\tau}_m\pi/2}\ket{\mathcal{E}(\psi)}$
    \State
    $g = \left|\frac{\text{DECODE}(\mathcal{E}, F, \ket{\psi^+}, \hat{H})-\text{DECODE}(\mathcal{E}, F, \ket{\psi^-}, \hat{H})}{2}\right|$
    \If{$g>g_{\text{max}}$}
    \State
    $L[0]\leftarrow \hat{\tau}_m$
    \State
    $g_{\text{max}}\leftarrow g$
    \EndIf
    \State
    $\text{APPEND}(L, (g_m, \hat{\tau}_m))$
    \EndFor
    \State
    \Return $L[0], g_{\text{max}}$
    \end{algorithmic}
\end{algorithm}

%yeh it will be important to write down the decoder as well. 

When combining RLE-FED with ADAPT-VQE, Eqn.~(\ref{eqn: gradient}) can be efficiently sampled. Building on the results from Ref. \cite{Qubit-Adapt-VQE}, the qubit operator pool ${\hat{\tau}_1,\hat{\tau}_2, ..., \hat{\tau}_n }$ scales linearly with the number of qubits $Q$. Under our scheme's qubit scaling, the size of this operator pool is effectively reduced from $M$ to $2N\log M$, potentially enabling VQE simulations involving a large number of fermionic modes. This reduction in operator pool size yields practical benefits, as demonstrated in our LiH simulation, where chemical accuracy is achieved with significantly fewer CNOT gates compared to the JW-encoded basis. Furthermore, fermionic observable decoding plays a key role in reducing quantum gate overhead, enhancing the feasibility of near-term applications.

\section{Numerical Details}

\label{appendix: numerical}
\subsection{The Compression Rate}
Numerical study of the compression rate $Q/M$ given $Q$. This table directly compares to the Graph-Based encoding table in \cite{bravyi2017tapering}, which shows qubit advantages in all settings. Meanwhile, the polylogarithmic encoding cannot be compared with because their encoding advantage only appears in the $10^5$ qubits regime. Logarithmic scaling in $M$ means that it is possible to compress a large number of fermionic modes with a qubit cost that is manageable in the FTQC regime. For instance, a 2-electron problem with 100 modes can be encoded with fewer than 20 qubits and a 4-electron problem with 400 modes in 50 qubits.

\begin{table}[!htp]
\resizebox{0.95\columnwidth}{!}{
\begin{tabular}{|c|c|c|c|c|c|c|c|c|c|c|c|}
\hline
\diagbox[dir=SE]{N}{M}{Q}  & 10 & 12 & 14 & 16 & 18 & 20  & 22  & 24  & 26  & 28  & 30  \\ \hline
2 & 22 & 36 & 48 & 64 & 90 & 118 & 158 & 226 & 316 & 420 & 580 \\ \hline
3 & 13 & 20 & 25 & 31 & 38 & 46  & 58  & 72  & 88  & 105 & 140 \\ \hline
4 & 11 & 14 & 18 & 23 & 27 & 31  & 36  & 42  & 50  & 60  & 71  \\ \hline
5 & 11 & 14 & 16 & 19 & 22 & 26  & 31  & 34  & 39  & 44  & 49  \\ \hline
6 & 11 & 13 & 16 & 18 & 21 & 23  & 27  & 31  & 34  & 37  & 41  \\ \hline
\end{tabular}
}
\caption{The numerical result of RLE compression. This table shows the maximum number of modes $M$ that we can encode given the number of electron $N$ (rows) and qubits $Q$ (columns) with RLE. All results were computed on an 8-core laptop.}\label{t1}
\end{table}

To analyze the scaling with respect to the number of modes, we evaluated the compression rate $Q/M$ and presented the results in Figure \ref{Fig1} for a 4-electron problem. We benchmarked our results against the JW basis, graph-based encoding, and the theoretical upper and lower bounds. Our compression rate for linear encoding surpasses that of all previous works [10] across all scenarios. Compared to the non-linear QEE scheme, we maintain $\mathcal{O}(N \log M)$ scaling while achieving polynomial measurement scaling.
\begin{figure}[!h] 
    \centering
    \includegraphics[width=0.46\textwidth]{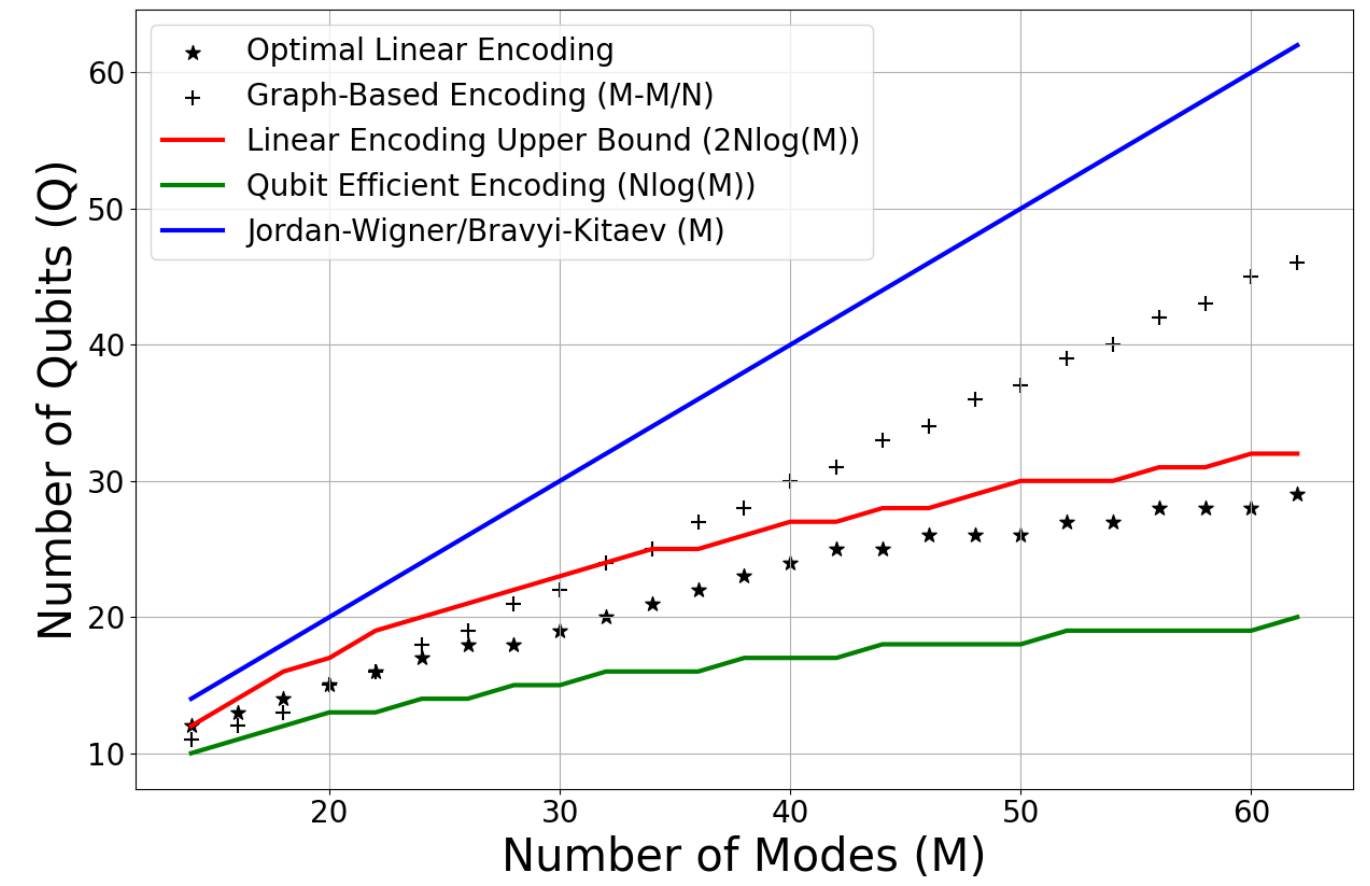}
    \caption{Comparison of the qubit cost of RLE (star) to the JW basis (blue), graph-based encoding (cross), and QEE (green) for a 4 electron problem and the $2N\log M$ bound (red) against the number of orbitals.}
    \label{Fig1}
\end{figure}

%We will need to change the naming of this section. Those referecing stuff I am not going to change in the arxiv version. 
\subsection{Hardware-Efficient Ansatz}
The Hardware-Efficient Ansatz (HEA) is commonly used in Variational Quantum Eigensolver (VQE) due to its resource efficiency given the current noisy constraints of quantum devices. Different hardware setups have their own specific HEA configurations, which include the selection of gates, circuit connectivity, and noise models. To meet these criteria, we design our HEA with two constraints. First, the gates used in our VQE simulations are restricted to the two-qubit Controlled-NOT gate and the single-qubit Y gate for rotations. Second, the connectivity of the two-qubit gates is constrained to a one dimensional chain topology. The design of the ansatz is therefore as follows:

$$U(\vec{\theta})=\Qcircuit @C=1.5em @R=0.7em {
&\gate{Ry(\theta_1)}& \ctrl{1} & \qw         & \qw         & \qw \gategroup{1}{1}{6}{6}{2em}{--} \\
&\gate{Ry(\theta_2)}& \targ    & \ctrl{1}    & \qw         & \qw \\
&\gate{Ry(\theta_3)}& \qw      & \targ       & \ctrl{1}    & \qw & \times n, \\
&\gate{Ry(\theta_4)}& \qw      & \qw         & \targ       & \ctrl{1}  \\
&\gate{Ry(\theta_5)}& \qw      & \qw         & \qw         & \targ   \\
& &\vdots&  &  &
}$$
the $n$ is the number of repetitions of the circuit block.

\subsection{Comparison with Unitary Couple Cluster type of Ansatz}
%We compared the performance of HEA ansatz in the compressed space to the UCC-type Ansatz \cite{sokolov2020quantum}, the golden standard to computational chemistry simulation, on ground state simulation of LiH in the STO-3G and 6-31G basis. With the active space chosen as $M=8$ and $M=16$ fermionic modes respectively, we can restrict the Hilbert space generated from both bases to the two electrons subspace -- denoted as (8,2) and (16, 2). In the STO-3G solution, both the 6-qubit compressed and 8-qubit uncompressed Hamiltonian can reach a chemical accuracy of 1 Kcal/mol using HEA with decent circuit depth. However, the 6-qubit Hamiltonian is 4 orders of magnitude more accurate compared to the 8-qubit Hamiltonian with less number of parameters and circuit depth. On the other hand, compared to HEA, the UCC-type ansatz converges faster and easier due to the small number of parameters, only at the expense of approximately 2-orders of magnitude deeper circuit and one order of magnitude the number of two-qubit gates to reach one order of magnitude higher accuracy of HEA results.

We compared the performance of the Hardware-Efficient Ansatz (HEA) in the compressed space to the UCC-type Ansatz \cite{sokolov2020quantum}, which is the gold standard for computational chemistry simulations, specifically for the ground state simulation of LiH in the STO-3G and 6-31G basis sets. For the active space chosen as $M=8$ and $M=16$ fermionic modes respectively, we restricted the Hilbert space generated from both bases to the subspace of two electrons, denoted as (8,2) and (16,2).

In the STO-3G case, both the 6-qubit compressed and 8-qubit uncompressed Hamiltonians achieved chemical accuracy of 1 Kcal/mol using HEA with a reasonable circuit depth. However, the 6-qubit Hamiltonian achieved accuracy four orders of magnitude higher compared to the 8-qubit Hamiltonian, with fewer parameters and shallower circuit depth.

In contrast, the UCC-type ansatz converges faster and more easily compared to HEA, owing to its smaller number of parameters. However, achieving approximately one order of magnitude higher accuracy than HEA requires a circuit depth approximately two orders of magnitude deeper and roughly ten times more two-qubit gates.

\begin{table}[!htb]
\begin{tabular}{|c|cccc|}
\hline
Systems & \multicolumn{4}{c|}{LiH (8,2)} \\ \hline
\multirow{2}{*}{Ansatz Types} & \multicolumn{2}{c|}{HEA} & \multicolumn{2}{c|}{UCC} \\ \cline{2-5} 
 & \multicolumn{1}{c|}{8 qubits} & \multicolumn{1}{c|}{6 qubits} & \multicolumn{1}{c|}{SUCCD} & UCCSD \\ \hline
Number of Parameters & \multicolumn{1}{c|}{48} & \multicolumn{1}{c|}{36} & \multicolumn{1}{c|}{6} & 15 \\ \hline
Number of CNOT & \multicolumn{1}{c|}{35} & \multicolumn{1}{c|}{25} & \multicolumn{1}{c|}{480} & 768 \\ \hline
Circuit Depth & \multicolumn{1}{c|}{41} & \multicolumn{1}{c|}{31} & \multicolumn{1}{c|}{668} & 1043 \\ \hline
$\Delta E$ (kcal/mol) & \multicolumn{1}{c|}{0.4126} & \multicolumn{1}{c|}{0.0002} & \multicolumn{1}{c|}{0.5216} & 0.00003 \\ \hline
\end{tabular}
\caption{This table showcases the noiseless simulation of (8,2) LiH using HEA and UCC-type ansatz on the state vector simulator. $\Delta E$ is the difference between the VQE energy and energy from exact diagonalisation. The HEA results are chosen from the lowest energy results using the VQE with $\texttt{L-BFGS-B}$ optimizer and 100 different initial states. The UCCSD and SUCCD ansatz is constructed from Qiskit package}\label{t2}
\end{table}
However, when scaling up to the (16,2) LiH system, the Hardware-Efficient Ansatz (HEA) applied directly to the uncompressed Hilbert space starts exhibiting Barren Plateaus. Even with L-BFGS-B optimization using 100 initial states, convergence to the required chemical accuracy is not achieved.

In contrast, using the RLE approach, which encodes the original 16-qubit Hamiltonian into an 8-qubit Hamiltonian, HEA can successfully converge to chemical accuracy.\begin{table}[!htb]
\begin{tabular}{|c|cccc|}
\hline
Systems & \multicolumn{4}{c|}{LiH (16,2)} \\ \hline
\multirow{2}{*}{Ansatz Types} & \multicolumn{2}{c|}{HEA} & \multicolumn{2}{c|}{UCC} \\ \cline{2-5} 
 & \multicolumn{1}{c|}{16 qubits} & \multicolumn{1}{c|}{8 qubits} & \multicolumn{1}{c|}{SUCCD} & UCCSD \\ \hline
Number of Parameters & \multicolumn{1}{c|}{96} & \multicolumn{1}{c|}{48} & \multicolumn{1}{c|}{28} & 63 \\ \hline
Number of CNOT & \multicolumn{1}{c|}{75} & \multicolumn{1}{c|}{35} & \multicolumn{1}{c|}{4032} & 7280 \\ \hline
Circuit Depth & \multicolumn{1}{c|}{81} & \multicolumn{1}{c|}{41} & \multicolumn{1}{c|}{4874} & 8675 \\ \hline
$\Delta E$ (kcal/mol) & \multicolumn{1}{c|}{2.2099} & \multicolumn{1}{c|}{0.7186} & \multicolumn{1}{c|}{0.5749} & 0.00001 \\ \hline
\end{tabular}
\caption{This table showcases the noiseless simulation of (16,2) LiH using HEA and UCC-type ansatz on the state vector simulator. $\Delta E$ is the difference between the VQE energy and energy from exact diagonalisation. The HEA results are chosen from the lowest energy results using the VQE with $\texttt{L-BFGS-B}$ optimizer and 100 different initial states. The UCCSD and SUCCD ansatz are constructed from the Qiskit package.}\label{t3}
\end{table}

In conclusion, the LiH system scaling result shows that the qubit compression can improve VQE convergence while reducing the qubit cost and circuit depth. %\cheng{We will need to reword this: Also, it is stated in \ref{} that the gradients of parameterized circuits vanish exponentially as qubit counts increase. Therefore, our encoding method would suffer less from the barren plateau problem comparing to common encoding schemes that require exponentially more qubits than our QEE method.}

\section{Randomized Linear Encoder}
\label{appendix: RLE}

\subsection{The algorithm}
RLE (Algorithm \ref{alg: RLE}) begins its random search by choosing a $Q$ within the bounded region $N\log_2 M < Q < 2N\log_2 M$. Then, the algorithm initializes the parity check matrix $\mathbf{G} = [I_Q | D]$ in the \textit{standard basis} \cite{err_ch1}, where $I_Q$ is a $Q \times Q$ identity matrix. The standard basis helps reduce the space of randomized search. We randomly generate the $Q \times (M-Q)$ matrix $D$ such that each column has even Hamming weight equal to the value $Q/2$, denoted as $\text{even}(Q/2)$. 

\begin{algorithm}[]
    \caption{Randomized Linear Encoder}
    \label{alg}
    \begin{algorithmic}[]
    \Require
      Target States $\mathbf{S}$; Testing subset $T_s$;
    \Ensure
      Codeword $\mathbf{C}$;
    \While{True}
    \State 
    Initialized Codeword $\mathbf{W}=[];$
    \State 
    Random generate D with constraints;
    \State
    $\mathbf{G}\leftarrow[I_Q|D]$;
    \State
    Create distance $2N+2$ error correction matrix $C\leftarrow[-D^T|I_{M-Q}]$;
    \State
    Create a checking list $w=[]$
        \For{$t_s$ in $T_s$}
            \State $w\leftarrow [w|D_H(\sum C^T(t_s) \: mod\:2,\mathbf{0})]$;
        \EndFor
        \If{$\min(w) \geq 2N+2$}
        \For{$\mathbf{s}$ in $\mathbf{S}$}
            \State $\mathbf{c}=\mathbf{C}\mathbf{s}$;
            \State $\mathbf{W}=[\mathbf{W}|\mathbf{c}]$;
        \EndFor
        \If{no repeat element in \textbf{W}}
        \State
        \Return $\mathbf{W}$
        \EndIf
        \EndIf
    \EndWhile
    \end{algorithmic}
    \label{alg: RLE}
\end{algorithm}

The second part of the algorithm consists of checks. Initially, the algorithm censors a large set of flawed generators by randomly verifying the Hamming weight of the element $\Vec{k}$ generated by the $M-Q$ dimensional vector $\Vec{v}$ and the CEC $\mathbf{C}$.
\begin{equation}
    \Vec{k} = \mathbf{C}\Vec{v}, \Vec{v}\in S
    \label{eqn: checking}
\end{equation}
where $S$ is defined as the set of all vectors with Hamming weight $2$ and $2K$ for $2K<N$, where $K$ is a randomly generated integer. The random check serves to maximize the probability of obtaining the correct encoder. Finally, the algorithm explicitly computes the states to check for repetitions. This process is repeated until the minimum $Q$ is found.

\begin{figure}[] 
    \centering
    \includegraphics[width=0.48\textwidth]{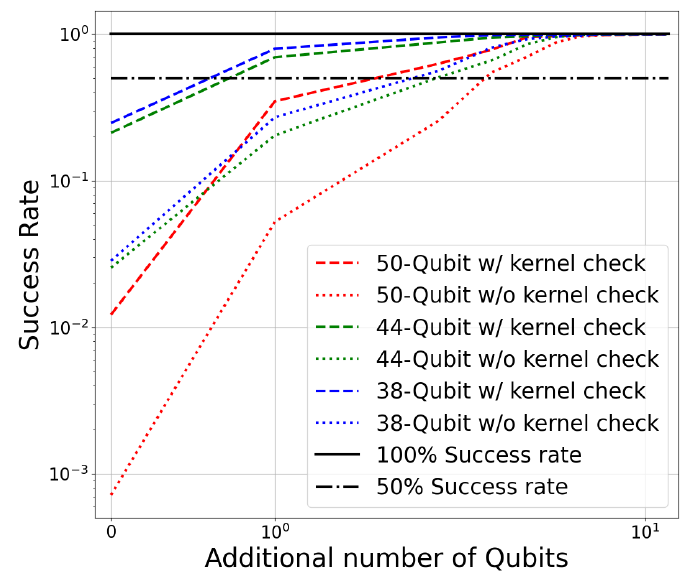}
    \caption{Impact on the success probability of generating $\mathbf{G}$ when adding extra qubits to the generator. The x-axis shows the number of auxiliary qubits and the y-axis is the success rate. The dashed line, which resulted from a codeword check generated in Eqn.~(\ref{eqn: checking}), shows a systematic improvement in success rate compared to the unchecked (dotted line) case.}
    \label{fig: FIG2}
\end{figure}

\subsection{Scaling of RLE}
We examine the scalability of RLE. The algorithm has a computation cost proportional to $\mathcal{O}(M^N)$ for each physical state check, multiplied by the number of checks performed. This complexity scaling is comparable to the Brouwer-Zimmermann algorithm, which is recognized as the most efficient method for checking the minimal distance of $\mathbf{C}$ \cite{Hernando_fast_BZ}. RLE also shares the same computational cost as non-linear encoding \cite{shee2022qubit}, but it is much simpler to implement. We can further enhance RLE scalability by incorporating auxiliary qubits into the optimal encoding. Figure \ref{fig: FIG2} summarizes the success probability of generating the correct parity check matrix $\mathbf{G}$ for 4-electron systems with auxiliary qubits. The algorithm's runtime exponentially converges to $\mathcal{O}(M^N)$ — the complexity for one check — when auxiliary qubits are successively added.

%should want to change the appendix

%yes, you will need to motivate use of the decoder. 

%we still need to 
%well, you will need to explain where A comes from. How is that related to the operator which you would like to decode. 

\end{document}